\documentclass{aptpub}
\authornames{Y.~BAI AND H.~LAM} 
\shorttitle{Confidence Intervals for Naive Rare-Event Estimators} 

\usepackage{latexsym}
\usepackage{graphicx}
\usepackage{caption}
\DeclareCaptionLabelFormat{andtable}{#1~#2  \&  \tablename~\thetable}
\usepackage{subcaption}
\usepackage{mathptmx}
\usepackage{amsmath}
\usepackage{amsfonts}
\usepackage{amssymb}
\usepackage{amsbsy}
\usepackage{algorithm}
\usepackage{algorithmic}
\usepackage[pdftex,colorlinks=true,urlcolor=blue,citecolor=black,anchorcolor=black,linkcolor=black]{hyperref}
\usepackage{appendix}
\usepackage[margin=1in]{geometry}
\usepackage{makecell}
\usepackage{enumerate}
\usepackage{breakcites}
\usepackage{xcolor}

\begin{document}

\title{Uncertainty Quantification and Confidence Intervals for Naive \\Rare-Event Estimators} 

\authorone[Columbia University]{Yuanlu Bai} 
\addressone{500 West 120th Street, New York, NY}
\emailone{yb2436@columbia.edu} 
\authortwo[Columbia University]{Henry Lam} 
\addresstwo{500 West 120th Street, New York, NY} 
\emailtwo{henry.lam@columbia.edu} 

\begin{abstract}
We consider the estimation of rare-event probabilities using sample proportions output by naive Monte Carlo or collected data. Unlike using variance reduction techniques, this naive estimator does not have a priori relative efficiency guarantee. On the other hand, due to the recent surge of sophisticated rare-event problems arising in safety evaluations of intelligent systems, efficiency-guaranteed variance reduction may face implementation challenges which, coupled with the availability of computation or data collection power, motivate the use of such a naive estimator. In this paper we study the uncertainty quantification, namely the construction, coverage validity and tightness of confidence intervals, for rare-event probabilities using only sample proportions. In addition to the known normality, Wilson's and exact intervals, we investigate and compare them with two new intervals derived from Chernoff's inequality and the Berry-Esseen theorem. Moreover, we generalize our results to the natural situation where sampling stops by reaching a target number of rare-event hits. Our findings show that the normality and Wilson's intervals are not always valid, but they are close to the newly developed valid intervals in terms of half-width. In contrast, the exact interval is conservative, but safely guarantees the attainment of the nominal confidence level. Our new intervals, while being more conservative than the exact interval, provide useful insights in understanding the tightness of the considered intervals.
\end{abstract}

\keywords{rare-event estimation; confidence interval; relative error; sample proportion}

\ams{62A01}{00A72}

\section{Introduction}
\label{sec:intro}
We consider the problem of estimating a minuscule probability, denoted $p=\mathbb{P}(A)$, for some rare event $A$, using data or Monte Carlo samples. This problem, known as rare-event estimation, is of wide interest to communities such as system reliability \cite{nicola1993fast,heidelberger1995fast,tuffin2004numerical,nicola2001techniques}, queueing systems \cite{sadowsky1991large,kroese1999efficient,blanchet2009rare,blanchet2014rare,szechtman2002rare,ridder2009importance,dupuis2009importance}, finance and insurance \cite{glasserman2004monte,glasserman2005importance,glasserman2008fast,mcneil2015quantitative,asmussen2010ruin,asmussen1985conjugate,collamore2002importance}, where it is crucial to estimate the likelihood of events which, though unlikely, can cause catastrophic impacts. 

There are multiple prominent lines of work addressing this estimation problem, depending on how information is collected. In settings where real-world data are collected, methods based on extreme value theory \cite{embrechts2013modelling,davison1990models,mcneil2015quantitative,smith1984threshold} are often used to extrapolate distributional tails to assist such estimation. These methods are theoretically justified and widely applicable, but their performance could be affected by intricate hyperparameter choices that affect their accuracy and challenge the reliability in uncertainty quantification \cite{embrechts2013modelling}. 
In settings where $A$ is an event described by a simulable model, Monte Carlo methods can be used, and to speed up computation one often harnesses variance reduction tools such as importance sampling \cite{siegmund1976importance,sadowsky1990large,juneja2006rare}, conditional Monte Carlo \cite{asmussen2007stochastic,rubinstein2016simulation} and multi-level splitting \cite{glasserman1999multilevel,villen1994restart,au2001estimation}. While variance reduction is greatly beneficial in reducing the number of Monte Carlo samples needed to estimate rare events \cite{bucklew2013introduction,asmussen2007stochastic,rubinstein2016simulation}, it is also widely known that they rely heavily on model assumptions \cite{juneja2006rare,blanchet2012state}. That is, to guarantee the successful performances of these techniques, one typically needs to analyze the underlying model dynamics carefully to design the Monte Carlo scheme. However, recent applications, such as autonomous vehicle safety evaluation \cite{zhao2016accelerated,zhao2017accelerated,huang2017accelerated,o2018scalable,arief2021deep} and robustness evaluation of machine learning predictors \cite{webb2018statistical,weng2018evaluating,bai2022rare}, lead to rare-event estimation problems with extremely sophisticated structures that hinder the design of efficiency-guaranteed variance reduction schemes. On the other hand, with the remarkable recent surge of computational infrastructure, in some situations one could afford to run gigantic amount of simulation trials.

Motivated by the limitations of the above techniques and the potential to generate numerous samples, in this paper we focus on a more basic setting than some of the above literature, but in a sense fundamental. More precisely, we focus on the situation where all we have to estimate $p$ is a set of i.i.d. Bernoulli observations $I(A)$. A natural point estimate of $p$ is the sample proportion $\hat p$, i.e., given a set of Bernoulli data $I_1,\ldots,I_n$ of size $n$, we output $\hat p=(1/n)\sum_{i=1}^nI_i$. We are interested in understanding the statistical error in using $\hat p$, in the situation where \emph{$p$ could be very small, importantly with no lower bound on how small it could be}. Unlike the estimates given by efficiency-guaranteed variance reduction techniques, as we will explain momentarily, it is not entirely straightforward whether using simple sample proportion can give meaningful guarantees to estimating rare-event probabilities, in relation to the sample size $n$ and the (unknown) magnitude of $p$. Motivated by this, our main goal of this paper is to study the construction, coverage validity and tightness of confidence intervals (CIs) for rare-event probabilities using only the simple sample proportion estimator. The main messages from our findings are as follows: The normality and Wilson's intervals are not always valid, in the sense that their actual coverage probabilities can be less than the nominal confidence level, but they are shown to be close to our two newly developed valid intervals in terms of half-width. On the other hand, the exact interval is conservative, as its coverage probability is strictly larger than the nominal confidence level and hence it is not as tight as the aforementioned two intervals, but it safely guarantees the attainment of the nominal confidence level. Our new intervals are even more conservative than the exact interval and hence not recommended in practice, but they provide useful insights in understanding the tightness of the normality and Wilson's intervals.


This paper is organized as follows. Section \ref{sec:problem} describes the problem setting and the motivating challenges. Section \ref{sec:overview CI} overviews the existing and new CIs, and Section \ref{sec:problem_conclusion} summarizes our main results. Then, in Sections \ref{sec:standard} and \ref{sec:stopping}, we present the details of the derivation and analyses of these intervals. After that, Section \ref{sec:numerical} reports some numerical results to visualize our comparisons. Section \ref{sec:conclusion} concludes this paper with our findings and recommendations. All missing proofs can be found in the appendix. 



\section{Problem Setting and Motivation}\label{sec:problem}


Suppose we would like to estimate a target probability $p$ by using information from the Bernoulli data, or equivalently $\hat p$. In particular, we would like to construct a CI for $p$ that has justifiable statistical guarantees. In answering this, we would also quantify the error between the point estimate $\hat p$ and $p$. 

First of all, we clarify what a good CI is supposed to be. To this end, we mainly consider the \emph{validity} of the coverage and \emph{tightness}. Throughout this paper, we say that $[\hat p_l(\alpha),\hat p_u(\alpha)]$ is a valid $(1-\alpha)$-level CI if $\mathbb{P}(\hat p_l(\alpha)\leq p\leq \hat p_u(\alpha))\geq 1-\alpha$. 
This notion of validity can be defined similarly for one-sided confidence bounds. On the other hand, a good CI should not be too wide; for example, in the extreme case, the trivial CI $[0,1]$ is valid, but it does not provide any useful information. In this paper, we quantify tightness by the ``half-width'', i.e., $\hat p_u(\alpha)-\hat p$ or $\hat p-\hat p_l(\alpha)$ (some intervals we consider are symmetric so there is no difference between the ``upper" and ``lower" half-widths, but some intervals are not, in which case the context would make the meaning of half-width clear). Importantly, considering that $p$ is tiny in the rare-event settings, the CI is meaningful only if the half-width is small relative to $p$ and $\hat p$.

To understand the challenges, we first examine the use of a standard ``textbook" CI, and we focus on the upper confidence bound for now since the lower confidence bound can be argued analogously. More specifically, we use the following as the $(1-\alpha)$-level upper confidence bound
\begin{equation}
\hat p^{CLT}=\hat p+z_{1-\alpha}\sqrt{\frac{\hat p(1-\hat p)}{n}}\label{naive}
\end{equation}
where $z_{1-\alpha}$ is the $(1-\alpha)$-quantile of a standard normal variable. The typical way to justify \eqref{naive} is a normal approximation using the central limit theorem (CLT), which entails that
\begin{equation}
\mathbb{P}(p\leq\hat p^{CLT})\approx\bar\Phi(-z_{1-\alpha})=1-\alpha\label{approx naive}
\end{equation}
where we denote $\bar\Phi$ (and $\Phi$) as the tail (and cumulative) distribution function of standard normal. 

To delve a little further, note that the approximation error in \eqref{approx naive} is controlled by the Berry-Esseen (B-E) Theorem. To simplify the discussion, suppose we are in a more idealized (but unrealistic) case that we know the precise value of the variance of the Bernoulli trial, i.e., $\sigma^2=p(1-p)$, so that we use $\hat p+z_{1-\alpha}\sigma/\sqrt n$. Then the B-E theorem stipulates that
\begin{equation}
|\mathbb{P}(p\leq\hat p^{CLT})-\Phi(z_{1-\alpha})|\leq\frac{C\rho}{\sigma^3\sqrt n}\label{Berry Essen}
\end{equation}
where $\rho=E|I_i-p|^3=p(1-p)(1-2p+2p^2)$, and $C$ is a universal constant ($\approx0.4748$). Thus, the error in \eqref{Berry Essen} is bounded by
\begin{equation}
\frac{Cp(1-p)(1-2p+2p^2)}{p^{3/2}(1-p)^{3/2}\sqrt n}\leq \frac{C}{\sqrt{np(1-p)}}.\label{bound transform}
\end{equation}
The issue is that when $p$ is tiny, $np$ can also be tiny unless $n$ is sufficiently big, but a priori we would not know what $n$ is ``sufficient". If we have used the confidence bound given by \eqref{naive} where the variance $\sigma^2$ is unknown and estimated by $\hat p(1-\hat p)$, a similar B-E bound would ultimately conclude the same issue as revealed by \eqref{bound transform} \cite{shao2013}. A straightforward idea is to use the number of successes to infer whether $n$ is sufficiently large. Suppose we have, say, 30 ``success" outcomes among $n$ trials, then we may think that $np\approx30$, so that from the bound \eqref{bound transform} the error of $\hat p^{CLT}$ appears controlled. As another more extreme case, suppose we only have only 1 success, then we may be led to believe $np\approx1$, so that $\hat p^{CLT}$ is well-defined but its coverage is likely way off from $1-\alpha$. However, we note that the guess that $np\approx30$ or $np\approx1$ is itself based on some central limit or concentration argument, which apparently leads to a circular reasoning. This challenge motivates us to investigate more on the validity and tightness of different CIs in order to make a suitable choice.

It is well known that a quick and implementable approach to construct a CI that is always valid regardless of $n$, $p$ or $\hat p$ is to utilize the fact that $n\hat p$ follows a binomial distribution and extract a finite-sample confidence region using this exact distribution. This is often called the Clopper-Pearson CI or the exact method \cite{clopper1934use}. Though this is computationally easy, we are interested in simpler mathematical forms that allow us to analytically study the relative half-width as well. In this regard, Wilson's interval \cite{agresti1998approximate} has been studied and shown to give superior empirical performances, even in the case that $p$ is tiny, but we are not aware of any rigorous proof on its validity. In this paper, we will propose two different ways of constructing CIs for $p$ that are simultaneously valid and analytically tractable, one using Chernoff's inequality, and the other one using the B-E bound. Compared to the exact CI, these two CIs have explicit forms that allow us to investigate their half-widths, thereby understand how far the CLT CI or Wilson's CI is from a valid CI.

Finally, in simulation analysis and some real data situations, it is natural to keep sampling until we observe enough successes (e.g., when the number of successes is 30) in the experiments. We will also adapt the existing or newly developed CIs to this setup and investigate their performance.

\section{Overview of Confidence Intervals}\label{sec:overview CI}

Here we briefly introduce the formulas of the CIs that we study in this paper. We will consider two settings respectively. The first one is called the ``standard'' setting, where the sample size $n$ is fixed. The other setting is when we fix the number of successes $\hat{s}=n\hat p$, which we call the ``targeted stopping'' setting. Under each setting, we will discuss three existing CIs: \emph{CLT CI}, \emph{Wilson's CI} and the \emph{Exact CI}. We will also introduce how to construct our new \emph{Chernoff's CI} and the \emph{B-E CI} via inverting Chernoff's inequality and the B-E theorem. 

\subsection{Confidence Intervals under Standard Setting}\label{sec:CI_standard}

Under the standard setting, to construct valid CIs, our beginning point is the following set:
\begin{equation}
    \{0< p< 1: F(\hat p)\geq\alpha/2,F_-(\hat p)\leq 1-\alpha/2\}
    \label{eqn:pvalue}
\end{equation}
where $F(x)=\mathbb{P}(\hat p\leq x)$ and $F_-(x)=\mathbb{P}(\hat p<x)$. Note that $F$ and $F_-$ depend on $p$. If $F$ were continuous, then we know that $\mathbb{P}(F(\hat p)\geq\alpha/2,F_-(\hat p)\leq 1-\alpha/2)=1-\alpha$ since in this case $F(\hat p)=F_-(\hat p)\stackrel{d}{=}Unif[0,1]$. Now we argue that $\mathbb{P}(F(\hat p)\geq\alpha/2,F_-(\hat p)\leq 1-\alpha/2)>1-\alpha$ in this discrete case. Indeed, for any $\alpha\in(0,1)$, there exist $0\leq k,l\leq n$ such that $F((k-1)/n)<\alpha/2\leq F(k/n)$ and $F_-(l/n)\leq 1-\alpha/2<F_-((l+1)/n)$. Then
\begin{align*}
    \mathbb{P}(F(\hat p)<\alpha/2\text{ or }F_-(\hat p)>1-\alpha/2)&\leq \mathbb{P}(F(\hat p)<\alpha/2)+\mathbb{P}(F_-(\hat p)>1-\alpha/2)\\
    &=\mathbb{P}(\hat p\leq (k-1)/n)+\mathbb{P}(\hat p\geq (l+1)/n)\\
    &=F((k-1)/n)+1-F_-((l+1)/n)<\alpha.
\end{align*}
Therefore, the set \eqref{eqn:pvalue} is a valid $(1-\alpha)$-level confidence region. From this derivation, we find that due to the discreteness, the probability that this confidence region covers the true value $p$ is strictly larger than the nominal confidence level $1-\alpha$, and hence this confidence region is inevitably conservative. 

The CLT CI and Wilson's CI can be obtained from \eqref{eqn:pvalue} by estimating $F(\hat p)$ and $F_-(\hat p)$ via normal approximation. As a result, these two CIs are no longer guaranteed to be valid. More specifically, using the fact that $\frac{\hat p-p}{\sqrt{\hat p(1-\hat p)/n}}\approx N(0,1)$, we substitute 
\begin{equation*}
F(\hat p)\approx \Phi\left(\frac{\hat p-p}{\sqrt{\hat p(1-\hat p)/n}}\right),F_-(\hat p)\approx \Phi\left(\frac{\hat p-p}{\sqrt{\hat p(1-\hat p)/n}}\right)
\end{equation*}
in \eqref{eqn:pvalue} to obtain the CLT CI:


\begin{definition}[CLT CI under Standard Setting]
Suppose that we estimate the probability $p=\mathbb{P}(A)$ for the event $A$, and $\hat p$ is the sample proportion of hitting $A$ in $n$ i.i.d. trials. Under this setting, the CLT CI is defined by:
\begin{align*}
    \hat{p}_u^{CLT}&=\hat p+z_{1-\alpha/2}\sqrt{\frac{\hat{p}(1-\hat{p})}{n}};\\
    \hat{p}_l^{CLT}&=\hat p-z_{1-\alpha/2}\sqrt{\frac{\hat{p}(1-\hat{p})}{n}}.
\end{align*}
\end{definition}

Similarly, by $\frac{\hat p-p}{\sqrt{p(1-p)/n}}\approx N(0,1)$ and substituting 
\begin{equation*}
F(\hat p)\approx \Phi\left(\frac{\hat p-p}{\sqrt{p(1- p)/n}}\right),F_-(\hat p)\approx \Phi\left(\frac{\hat p-p}{\sqrt{p(1-p)/n}}\right)
\end{equation*}
we get the Wilson's CI:
\begin{definition}[Wilson's CI under Standard Setting]
Suppose that we estimate the probability $p=\mathbb{P}(A)$ for the event $A$, and $\hat p$ is the sample proportion of hitting $A$ in $n$ i.i.d. trials. Under this setting, the Wilson's CI is defined by:
\begin{align*}
    \hat{p}_u^{Wilson}&=\frac{1+\frac{2n\hat{p}}{z_{1-\alpha/2}^2}+\sqrt{1+\frac{4n\hat{p}(1-\hat{p})}{z_{1-\alpha/2}^2}}}{2\left(1+\frac{n}{z_{1-\alpha/2}^2}\right)};\\
    \hat{p}_l^{Wilson}&=\frac{1+\frac{2n\hat{p}}{z_{1-\alpha/2}^2}-\sqrt{1+\frac{4n\hat{p}(1-\hat{p})}{z_{1-\alpha/2}^2}}}{2\left(1+\frac{n}{z_{1-\alpha/2}^2}\right)}.
\end{align*}
\end{definition}

Instead of using normal approximation, the exact CI directly solves the valid confidence region \eqref{eqn:pvalue}. In fact, we know that $\hat{s}=n\hat p\sim Binomial(n, p)$, so the functions $F(\cdot)$ and $F_-(\cdot)$ has exact expressions. More specifically, we have the following definition:

\begin{definition}[Exact CI under Standard Setting]
Suppose that we estimate the probability $p=\mathbb{P}(A)$ for the event $A$, and $\hat p$ is the sample proportion of hitting $A$ in $n$ i.i.d. trials. Under this setting, the Exact CI is defined by $[\hat{p}^{Exact}_l,\hat{p}^{Exact}_u]$ where $\hat{p}^{Exact}_u$ and $\hat{p}^{Exact}_l$ are respectively the solution to 
\begin{align*}
    \sum_{k=0}^{\hat{s}}\binom{n}{k}p^k(1-p)^{n-k}&=\alpha/2,\\
    \sum_{k=\hat{s}}^n\binom{n}{k}p^k(1-p)^{n-k}&=\alpha/2,
\end{align*}
except that $\hat{p}^{Exact}_u=1$ if $\hat{s}=n$ and $\hat{p}^{Exact}_l=0$ if $\hat{s}=0$. 
\end{definition}
When $0<\hat{s}<n$, the bounds could be expressed explicitly via quantiles of the $F$ distribution or Beta distribution, and hence are easy to compute numerically \cite{agresti1998approximate}. However, it is hard to analyze the scale of this CI, which motivates us to further relax the confidence region \eqref{eqn:pvalue} to get other valid CIs which are more conservative but easier to analyze.

In order to relax \eqref{eqn:pvalue}, we respectively consider using two methods: Chernoff's inequality and the B-E theorem. We will only present the formulas for them here and leave the details of their development to Section \ref{sec:standard_develop}.

\begin{definition}[Chernoff's CI under Standard Setting]
Suppose that we estimate the probability $p=\mathbb{P}(A)$ for the event $A$, and $\hat p$ is the sample proportion of hitting $A$ in $n$ i.i.d. trials. Under this setting, the Chernoff's CI is defined by:
\begin{align*}
    \hat{p}^{Chernoff}_u&=\hat{p}+\frac{\log(2/\alpha)}{n}+\sqrt{\frac{(\log(2/\alpha))^2}{n^2}+\frac{2\hat{p}\log(2/\alpha)}{n}},\\
    \hat{p}^{Chernoff}_l&=\hat{p}+\frac{\log(2/\alpha)}{2n}-\sqrt{\frac{(\log(2/\alpha))^2}{4n^2}+\frac{2\hat{p}\log(2/\alpha)}{n}}.
\end{align*}
\end{definition}
\begin{definition}[B-E CI under Standard Setting]
Suppose that we estimate the probability $p=\mathbb{P}(A)$ for the event $A$, and $\hat p$ is the sample proportion of hitting $A$ in $n$ i.i.d. trials. In addition, we assume that $p<\frac12$. Under this setting, the B-E CI is solved from 
\begin{align*}
    &\left\{0<p\leq \hat p\wedge\frac12:\Phi\left(\frac{p-\hat{p}}{\sqrt{p(1-p)/n}}\right)+\frac{C}{\sqrt{np(1-p)}}\geq\frac{\alpha}{2}\right\}\\
    \cup&\left\{\hat p\leq p<\frac12:\Phi\left(\frac{\hat{p}-p}{\sqrt{p(1-p)/n}}\right)+\frac{C}{\sqrt{np(1-p)}}\geq\frac{\alpha}{2}\right\}.
\end{align*}
\end{definition}

\subsection{Confidence Intervals under Targeted Stopping}\label{sec:CI_targeted}

Now we consider experiments where we keep sampling until we get $n_0$ successes. Under this setting, the sample size $N$ is a random variable. More specifically, $N=N_1+\cdots+N_{n_0}$ where $N_1,\cdots,N_{n_0}$ are i.i.d. $Geometric(p)$ random variables, or equivalently, $N-n_0$ follows a negative binomial distribution $NB(n_0,p)$. Note that $N\ge n_0$. We define $F_N(x)=\mathbb{P}(N\leq x)$ and $F_{N-}(x)=\mathbb{P}(N<x)$. Similar to Section \ref{sec:CI_standard}, we will argue that the following set is a valid $(1-\alpha)$-level confidence region for $p$:
\begin{equation}
    \{0< p< 1: F_N(N)\geq\alpha/2,F_{N-}(N)\leq 1-\alpha/2\}.
    \label{eqn:pvalueN}
\end{equation}
Indeed, for any $\alpha\in(0,1)$, there exist $1\leq k,l<\infty$ such that $F_N(k-1)<\alpha/2\leq F_N(k)$ and $F_{N-}(l)\leq 1-\alpha/2<F_{N-}(l+1)$. Then
\begin{align*}
    \mathbb{P}(F_N(N)<\alpha/2\text{ or }F_{N-}(N)>1-\alpha/2)&\leq \mathbb{P}(F_N(N)<\alpha/2)+\mathbb{P}(F_{N-}(N)>1-\alpha/2)\\
    &=\mathbb{P}(N\leq k-1)+\mathbb{P}(N\geq l+1)\\
    &=F_N(k-1)+1-F_{N-}(l+1)<\alpha.
\end{align*}
By definition, the set \eqref{eqn:pvalueN} is a valid $(1-\alpha)$-level confidence region.

We could still use the CLT CI and Wilson's CI with $\hat p=n_0/N$. More specifically, the CLT CI is 
\begin{definition}[CLT CI under Targeted Stopping]
Suppose that we estimate the probability $p=\mathbb{P}(A)$ for the event $A$. We keep sampling until we get $n_0$ successes and the sample size is denoted by $N$. Under this setting, the CLT CI is defined by:
\begin{align*}
    \hat{p}_{u,n_0}^{CLT}&=\frac{n_0}{N}+z_{1-\alpha/2}\sqrt{\frac{n_0(N-n_0)}{N^3}};\\
    \hat{p}_{l,n_0}^{CLT}&=\frac{n_0}{N}-z_{1-\alpha/2}\sqrt{\frac{n_0(N-n_0)}{N^3}}.
\end{align*}
\end{definition}
The Wilson's CI is 
\begin{definition}[Wilson's CI under Targeted Stopping]
Suppose that we estimate the probability $p=\mathbb{P}(A)$ for the event $A$. We keep sampling until we get $n_0$ successes and the sample size is denoted by $N$. Under this setting, the Wilson's CI is defined by:
\begin{align*}
    \hat{p}_{u,n_0}^{Wilson}&=\frac{1+\frac{2n_0}{z_{1-\alpha/2}^2}+\sqrt{1+\frac{4n_0(N-n_0)}{z_{1-\alpha/2}^2N}}}{2\left(1+\frac{N}{z_{1-\alpha/2}^2}\right)};\\
    \hat{p}_{l,n_0}^{Wilson}&=\frac{1+\frac{2n_0}{z_{1-\alpha/2}^2}-\sqrt{1+\frac{4n_0(N-n_0)}{z_{1-\alpha/2}^2N}}}{2\left(1+\frac{N}{z_{1-\alpha/2}^2}\right)}.
\end{align*}
\end{definition}

Similar to the standard setting, we can directly solve \eqref{eqn:pvalueN} using the exact distribution of $N$. 
\begin{definition}[Exact CI under Targeted Stopping]
Suppose that we estimate the probability $p=\mathbb{P}(A)$ for the event $A$. We keep sampling until we get $n_0$ successes and the sample size is denoted by $N$. Under this setting, the Exact CI is defined by $[\hat{p}_{l,n_0}^{Exact},\hat{p}_{u,n_0}^{Exact}]$ where $\hat{p}_{u,n_0}^{Exact}$ and $\hat{p}_{l,n_0}^{Exact}$ are respectively the solution to
\begin{align*}
    \sum_{k=0}^{N-n_0-1}\binom{k+n_0-1}{n_0-1}(1-p)^kp^{n_0}&=1-\alpha/2,\\
    \sum_{k=0}^{N-n_0}\binom{k+n_0-1}{n_0-1}(1-p)^kp^{n_0}&=\alpha/2,
\end{align*}
except that $\hat{p}_{u,n_0}^{Exact}=1$ if $N=n_0$. 
\end{definition}
While the interval is easy to compute numerically, it is not easy to analyze. Similar to the standard setting, we will relax the confidence region \eqref{eqn:pvalueN} to construct valid CIs respectively via inverting the Chernoff's inequality and the B-E theorem. We leave the details of developing these two new CIs to Section \ref{sec:stopping_develop} and only present the formulas here. 

\begin{definition}[Chernoff's CI under Targeted Stopping]
Suppose that we estimate the probability $p=\mathbb{P}(A)$ for the event $A$. We keep sampling until we get $n_0$ successes and the sample size is denoted by $N$. Under this setting, the Chernoff's CI is solved from
\begin{equation*}
    \left\{0<p<1:p^{n_0}(1-p)^{N-n_0}\geq \frac{\alpha}{2}\left(\frac{n_0}{N}\right)^{n_0}\left(1-\frac{n_0}{N}\right)^{N-n_0}\right\}.
\end{equation*}
\end{definition}
\begin{definition}[B-E CI under Targeted Stopping]
Suppose that we estimate the probability $p=\mathbb{P}(A)$ for the event $A$. We keep sampling until we get $n_0$ successes and the sample size is denoted by $N$. In addition, we assume that $p<\frac12$. Under this setting, the B-E CI is defined by:
\begin{align*}
    &\left\{0< p\leq\frac{n_0}{N}\wedge\frac12:\Phi\left(\frac{Np-n_0}{\sqrt{n_0(1-p)}}\right)+\frac{C'}{\sqrt{n_0(1-p)^3}}\geq\frac{\alpha}{2}\right\}\\
    \cup&\left\{\frac{n_0}{N}\leq p<\frac{1}{2}:\Phi\left(\frac{n_0-Np}{\sqrt{n_0(1-p)}}\right)+\frac{C'}{\sqrt{n_0(1-p)^3}}\geq\frac{\alpha}{2}\right\}
\end{align*}
where $C'=16C$ is a universal constant.
\end{definition}

\section{Summary of Main Results}\label{sec:problem_conclusion}

As explained in Section \ref{sec:problem}, we mainly consider the validity in terms of coverage probability and tightness in terms of half-width when we compare different CIs. In terms of validity, the existing CLT CI and Wilson's CI do not possess guarantees, while the Exact CI, and our new Chernoff's CI and B-E CI are valid by construction. The half-widths of the CIs, which will be analyzed in detail in Sections \ref{sec:standard_width} and \ref{sec:stopping_width}, can be summarized in Figures \ref{fig:comparison_standard} and \ref{fig:comparison_targeted}. In particular, Figures \ref{fig:comparison_standard} and \ref{fig:comparison_targeted} illustrate the comparisons of these CIs in terms of upper and lower bound (the exact CI is not included since it is hard to analyze its magnitude). For instance, from \eqref{naive}, we clearly see that the half-width of the CLT CI scales in the same order as $\sqrt{\hat p/n}=\hat p/\sqrt{\hat{s}}$ where $\hat{s}=n\hat p$ is the number of positive outcomes. 
By expressing $\sqrt{\hat p/n}$ as $\hat p/\sqrt{\hat{s}}$ here, 
it is easier to see how the half-width scales relative to $\hat p$. That is, the relative half-width is of order $1/\sqrt{\hat{s}}$. Note that in these figures, under the standard setting, for $f(n,\hat p),g(n,\hat p)\geq 0$, we write $f=O(g)$ if there exist $N_0,p_0,M>0$ which do not depend on $p$ or $\hat p$ such that for any $n\hat p>N_0$ and $\hat p<p_0$, $f\leq Mg$; we write $f=\Theta(g)$ if there exist $N_0,p_0,M_1,M_2>0$ which do not depend on $p$ or $\hat p$ such that for any $n\hat p>N_0$ and $\hat p<p_0$, $M_1g\leq f\leq M_2g$. Under the targeted stopping setting, $O(\cdot)$ and $\Theta(\cdot)$ are defined similarly by replacing $n\hat p$ with $n_0$. The notations $O(\cdot)$ and $\Theta(\cdot)$ will be used throughout the rest of this paper.

\begin{figure}[!ht]
\centering
    \begin{subfigure}{0.4\textwidth}
        \centering
        \includegraphics[width=\textwidth]{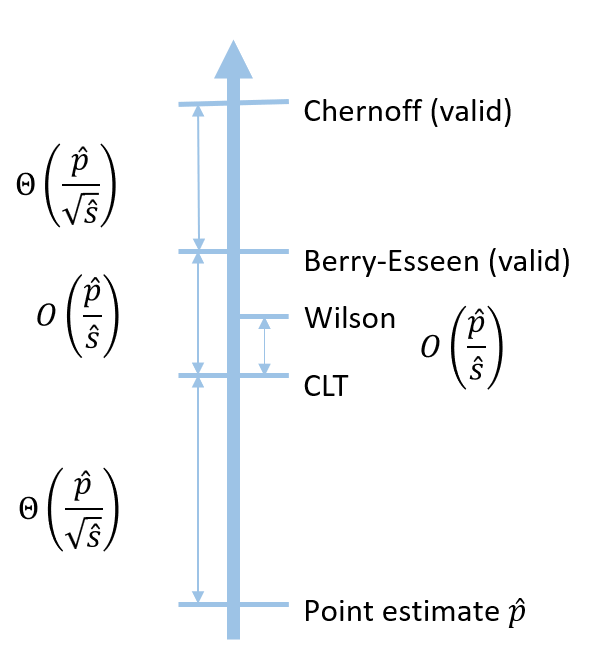}
        \caption{Standard - Upper}
    \end{subfigure}
    \begin{subfigure}{0.4\textwidth}
        \centering
        \includegraphics[width=\textwidth]{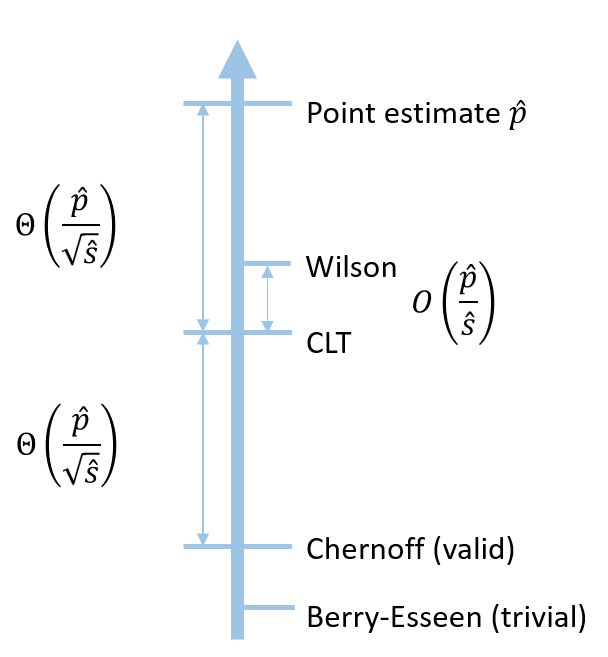}
        \caption{Standard - Lower}
    \end{subfigure}
    \caption{Comparisons of the positions of confidence upper and lower bounds under the standard setting. In this figure, ``valid'' means that the CI has valid coverage in the sense that the actual coverage probability always reaches the nominal confidence level.} 
    \label{fig:comparison_standard}
\end{figure}

\begin{figure}[!ht]
\centering
    \begin{subfigure}{0.4\textwidth}
        \centering
        \includegraphics[width=\textwidth]{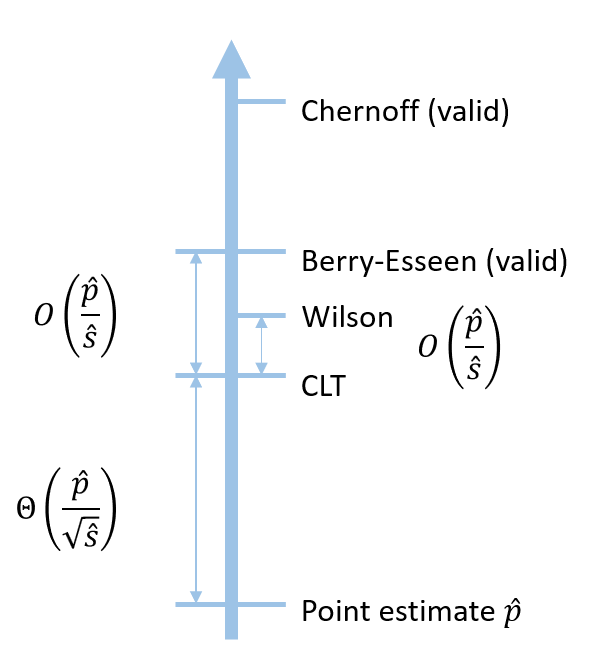}
        \caption{Targeted Stopping - Upper}
    \end{subfigure}
    \begin{subfigure}{0.4\textwidth}
        \centering
        \includegraphics[width=\textwidth]{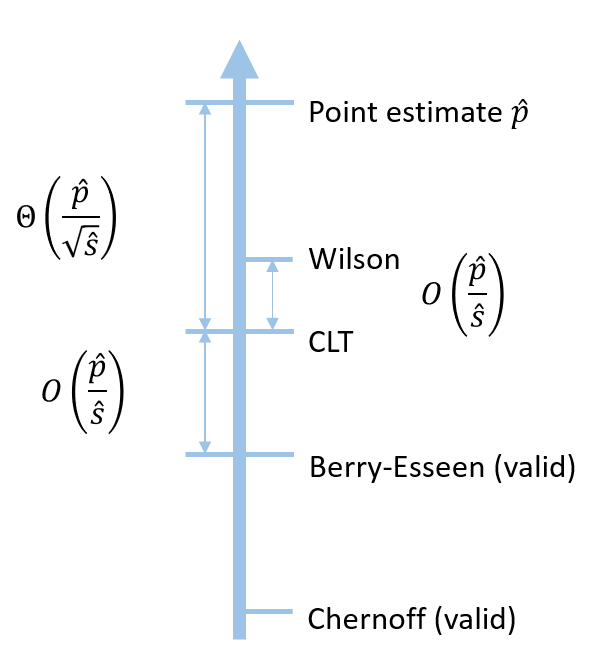}
        \caption{Targeted Stopping - Lower}
    \end{subfigure}
    \caption{Comparisons of the positions of confidence upper and lower bounds under the targeted stopping setting. In this figure, ``valid'' means that the CI has valid coverage in the sense that the actual coverage probability always reaches the nominal confidence level.} 
    \label{fig:comparison_targeted}
\end{figure}

More concretely, Table \ref{tab:summary} summarizes the formula, scale, pros and cons of each CI under both settings. The key findings are summarized as follows:
\begin{itemize}
    \item The CLT CI is the ``textbook" normality interval and thus very intuitive, but its coverage probability can be far below the nominal level. However, in terms of the half-width, except for the lower bound in the standard setting, the difference between the CLT bound and the valid B-E bound is of order $\hat p/\hat{s}$, so the relative difference with respect to $\hat p$ is of order $1/\hat{s}$, which is of higher order in $\hat{s}$ than its relative half-width. This can be viewed as a relatively small price of validity paid to make the CLT bound correct. For the lower bound in the standard setting, the B-E bound is trivial, so we cannot come to a similar conclusion. However, in this case the difference between the CLT bound and the valid Chernoff's bound is of order $\hat p/\sqrt{\hat{s}}$, the same order as the half-width, which shows that the CLT bound roughly has a correct magnitude.
    \item In practice, Wilson's CI has a satisfactory performance, in the sense that it is relatively tight while the coverage probability is usually close to the nominal confidence level. The difference between Wilson's bound and the CLT bound is of order $\hat p/\hat{s}$, which is of higher order in $\hat{s}$ than the half-width. As a result, the conclusions for the CLT CI regarding the difference from the valid B-E and Chernoff's CIs still hold for the Wilson's CI.
    \item The Exact CI is, as aforementioned, inevitably conservative, in the sense that its coverage probability is strictly higher than the nominal level. However, it is the tightest among the valid CIs, so it is recommended when one wants the nominal confidence level to be guaranteed. Chernoff's and the B-E CIs are valid but extremely conservative. They are not recommended to use in practice, but their analytical forms help us gain useful insights on the CLT and Wilson's CIs. That is, now we learn that the CLT and Wilson's CIs, althought not always valid, are relatively close to these two valid CIs as mentioned in the first bulletpoint.
\end{itemize}

\begin{table}[!ht]
    \begin{subtable}{\textwidth}
        \centering
    \resizebox{\columnwidth}{!}{\begin{tabular}{|l|c|l|}
    \hline
    \multicolumn{3}{|c|}{Standard Setting ($I_1,\dots,I_n\stackrel{i.i.d.}{\sim}Bernoulli(p),\hat{p}=\frac{1}{n}\sum_{i=1}^n I_i,\hat{s}=n\hat p$)}\\
    \hline
    CLT & $\hat p\pm z_{1-\alpha/2}\sqrt{\frac{\hat{p}(1-\hat{p})}{n}}.$&\makecell*[l]{
    Scale: $\hat{p}_u^{CLT}-\hat{p}=\hat{p}-\hat{p}_l^{CLT}=\Theta(\hat p/\sqrt{\hat{s}})$;\\
    Pros: Follows from the intuitive ``textbook" formula;\\
    Cons: Not always valid, especially when $np$ is not sufficiently large.}
    \\
    \hline
    Wilson &$\frac{1+\frac{2n\hat{p}}{z_{1-\alpha/2}^2}\pm\sqrt{1+\frac{4n\hat{p}(1-\hat{p})}{z_{1-\alpha/2}^2}}}{2\left(1+\frac{n}{z_{1-\alpha/2}^2}\right)}.$ &
    \makecell*[l]{Scale: $|\hat{p}_u^{Wilson}-\hat{p}_u^{CLT}|=O(\hat p/\hat{s}),|\hat{p}_l^{Wilson}-\hat{p}_l^{CLT}|=O(\hat p/\hat{s})$;\\
    Pros: Tight, the coverage probability is usually close to $1-\alpha$;\\
    Cons: Not always valid, lacks theoretical error control.}\\
    \hline
    Exact &\makecell*[c]{Solutions to \\$\sum_{k=0}^{\hat{s}}\binom{n}{k}p^k(1-p)^{n-k}=\alpha/2,$\\
    $\sum_{k=\hat{s}}^n\binom{n}{k}p^k(1-p)^{n-k}=\alpha/2$\\
    except that $\hat{p}^{Exact}_u=1$ if $\hat{s}=n$ and $\hat{p}^{Exact}_l=0$ if $\hat{s}=0$.}& \makecell*[l]{
    Pros: Always valid, tighter than other valid CI;\\
    Cons: Conservative, hard to analyze.}\\
    \hline
    Chernoff &\makecell*[c]{$\hat{p}+\frac{\log(2/\alpha)}{n}+\sqrt{\frac{(\log(2/\alpha))^2}{n^2}+\frac{2\hat{p}\log(2/\alpha)}{n}},$\\
    $\hat{p}+\frac{\log(2/\alpha)}{2n}-\sqrt{\frac{(\log(2/\alpha))^2}{4n^2}+\frac{2\hat{p}\log(2/\alpha)}{n}}.$}& \makecell*[l]{Scale: $\hat{p}_u^{Chernoff}-\hat{p}_u^{CLT}=\Theta(\hat p/\sqrt{\hat{s}}),\hat{p}_l^{CLT}-\hat{p}_l^{Chernoff}=\Theta(\hat p/\sqrt{\hat{s}});$\\
    Pros: Always valid, helps us understand the relative error of $\hat p$;\\
    Cons: Extremely conservative.}\\
    \hline
    B-E & \makecell*[c]{$\left\{0<p\leq \hat p\wedge\frac12:\Phi\left(\frac{p-\hat{p}}{\sqrt{p(1-p)/n}}\right)+\frac{C}{\sqrt{np(1-p)}}\geq\frac{\alpha}{2}\right\}$\\$\cup\left\{\hat p\leq p<\frac12:\Phi\left(\frac{\hat{p}-p}{\sqrt{p(1-p)/n}}\right)+\frac{C}{\sqrt{np(1-p)}}\geq\frac{\alpha}{2}\right\}$\\where $C$ is the universal constant in the B-E theorem.}&\makecell*[l]{Scale: $|\hat{p}_u^{BE}-\hat{p}_u^{CLT}|=O(\hat p/\hat{s})$;\\
    Pros: Always valid, helps us understand the error of the CLT upper bound;\\
    Cons: Extremely conservative, trivial lower bound.}\\
    \hline
    \end{tabular}}
    \end{subtable}
    \\
    \par\bigskip
    \begin{subtable}{\textwidth}
        \centering
    \resizebox{\columnwidth}{!}{\begin{tabular}{|l|c|l|}
    \hline
    \multicolumn{3}{|c|}{Targeted Stopping Setting ($N_1,\dots,N_{n_0}\stackrel{i.i.d.}{\sim}Geometric(p),N=\sum_{i=1}^{n_0}N_i,\hat p=n_0/N$)}\\
    \hline
    CLT &$\frac{n_0}{N}\pm z_{1-\alpha/2}\sqrt{\frac{n_0(N-n_0)}{N^3}}.$&\makecell*[l]{
    Scale: $\hat{p}_{u,n_0}^{CLT}-\hat{p}=\hat{p}-\hat{p}_{l,n_0}^{CLT}=\Theta(\sqrt{n_0}/N)$;\\
    Pros: Follows from the intuitive ``textbook" formula;\\
    Cons: Not always valid.}\\
    \hline
    Wilson &$\frac{1+\frac{2n_0}{z_{1-\alpha/2}^2}\pm\sqrt{1+\frac{4n_0(N-n_0)}{z_{1-\alpha/2}^2N}}}{2\left(1+\frac{N}{z_{1-\alpha/2}^2}\right)}.$&\makecell*[l]{Scale: $|\hat{p}_{u,n_0}^{Wilson}-\hat{p}_{u,n_0}^{CLT}|=O(1/N),|\hat{p}_{l,n_0}^{Wilson}-\hat{p}_{l,n_0}^{CLT}|=O(1/N)$;\\
    Pros: Tight, the coverage probability is usually close to $1-\alpha$;\\
    Cons: Not always valid.}\\
    \hline
    Exact &\makecell*[c]{Solutions to\\$\sum_{k=0}^{N-n_0-1}\binom{k+n_0-1}{n_0-1}(1-p)^kp^{n_0}=1-\alpha/2$,\\
    $\sum_{k=0}^{N-n_0}\binom{k+n_0-1}{n_0-1}(1-p)^kp^{n_0}=\alpha/2$\\
    except that $\hat{p}_{u,n_0}^{Exact}=1$ if $N=n_0$.}& \makecell*[l]{Pros: Always valid, tighter than other valid CI;\\
    Cons: Conservative, hard to analyze.}\\
    \hline
    Chernoff &$\left\{0<p<1:p^{n_0}(1-p)^{N-n_0}\geq \frac{\alpha}{2}\left(\frac{n_0}{N}\right)^{n_0}\left(1-\frac{n_0}{N}\right)^{N-n_0}\right\}.$& \makecell*[l]{Pros: Always valid; \\ Cons: Extremely conservative, hard to analyze.}\\
    \hline
    B-E &\makecell*[c]{$\left\{0< p\leq\frac{n_0}{N}\wedge\frac12:\Phi\left(\frac{Np-n_0}{\sqrt{n_0(1-p)}}\right)+\frac{C'}{\sqrt{n_0(1-p)^3}}\geq\frac{\alpha}{2}\right\}$\\$\cup\left\{\frac{n_0}{N}\leq p<\frac{1}{2}:\Phi\left(\frac{n_0-Np}{\sqrt{n_0(1-p)}}\right)+\frac{C'}{\sqrt{n_0(1-p)^3}}\geq\frac{\alpha}{2}\right\}$\\
    where $C'=16C$ is a universal constant.}&\makecell*[l]{Scale: $\hat{p}_{u,n_0}^{BE}-\hat{p}_{u,n_0}^{CLT}=O(1/N), \hat{p}_{l,n_0}^{CLT}-\hat{p}_{l,n_0}^{BE}=O(1/N)$;\\
    Pros: Aways valid, helps us understand the error of $\hat p$ and the CLT CI;\\
    Cons: Extremely conservative, trivial for small $n_0$.} \\
    \hline
    \end{tabular}}
    \end{subtable}
    \caption{Summary of the CIs.}
    \label{tab:summary}
\end{table}

\section{Developments under Standard Setting}\label{sec:standard}
We present in detail the construction of the two new CIs, Chernoff's CI and the B-E CI, that endows their validity (Section \ref{sec:standard_develop}). Then we analyze the half widths of all our discussed CIs (Section \ref{sec:standard_width}).

\subsection{Derivation of New Confidence Intervals}\label{sec:standard_develop}

\paragraph{Chernoff's CI.}

Now we present our first approach to construct a valid CI for $p$ by relaxing \eqref{eqn:pvalue}. By Chernoff's inequality, we have
\begin{align*}
    \mathbb{P}(\hat{p}\leq(1-\delta)p)&\leq\exp\left(-\frac{\delta^2}{2}np\right),0<\delta<1;\\
    \mathbb{P}(\hat{p}\geq(1+\delta)p)&\leq\exp\left(-\frac{\delta^2}{2+\delta}np\right),\delta>0.
\end{align*}
Replacing $(1-\delta)p$ or $(1+\delta)p$ by $x$, we have
\begin{align*}
    F(x)&\leq \exp\left\{-\left(1-\frac{x}{p}\right)^2\frac{np}{2}\right\},x\leq p;\\
    F_-(x)&\geq 1-\exp\left\{-\frac{\left(\frac{x}{p}-1\right)^2}{1+\frac{x}{p}}np\right\},x\geq p.
\end{align*}
Hence $F(\hat{p})\geq\alpha/2,F_-(\hat{p})\leq1-\alpha/2$ implies that either 
\begin{equation*}
    p\geq\hat{p}\text{ and }\exp\left\{-\left(1-\frac{\hat p}{p}\right)^2\frac{np}{2}\right\}\geq\alpha/2
\end{equation*}
or 
\begin{equation*}
    p\leq\hat{p}\text{ and }1-\exp\left\{-\frac{\left(\frac{\hat p}{p}-1\right)^2}{1+\frac{\hat p}{p}}np\right\}\leq 1-\alpha/2.
\end{equation*}
Therefore, 
\begin{equation*}
    \left\{0< p\leq \hat p:\exp\left\{-\frac{\left(\frac{\hat p}{p}-1\right)^2}{1+\frac{\hat p}{p}}np\right\}\geq\alpha/2\right\}\cup\left\{\hat{p}\leq p< 1:\exp\left\{-\left(1-\frac{\hat p}{p}\right)^2\frac{np}{2}\right\}\geq\alpha/2\right\}
\end{equation*}
is a confidence region for $\hat p$ with confidence level at least $1-\alpha$. Simplifying the expression above, we have that 
\begin{align*}
    0< p\leq \hat p,\exp\left\{-\frac{\left(\frac{\hat p}{p}-1\right)^2}{1+\frac{\hat p}{p}}np\right\}\geq\alpha/2
    \Rightarrow\hat{p}+\frac{\log(2/\alpha)}{2n}-\sqrt{\frac{(\log(2/\alpha))^2}{4n^2}+\frac{2\hat{p}\log(2/\alpha)}{n}}\leq p\leq \hat p
\end{align*}
and
\begin{align*}
    \hat{p}\leq p< 1,\exp\left\{-\left(1-\frac{\hat p}{p}\right)^2\frac{np}{2}\right\}\geq\alpha/2\Rightarrow \hat p\leq p\leq \hat{p}+\frac{\log(2/\alpha)}{n}+\sqrt{\frac{(\log(2/\alpha))^2}{n^2}+\frac{2\hat{p}\log(2/\alpha)}{n}}.
\end{align*}
Hence, by taking the union, we get a valid $(1-\alpha)$-level CI for $p$, for any finite sample $n$. This can be summarized as:
\begin{theorem}[Validity of Chernoff's CI under Standard Setting]
The interval given by 
\begin{align*}
    \hat{p}^{Chernoff}_u&=\hat{p}+\frac{\log(2/\alpha)}{n}+\sqrt{\frac{(\log(2/\alpha))^2}{n^2}+\frac{2\hat{p}\log(2/\alpha)}{n}},\\
    \hat{p}^{Chernoff}_l&=\hat{p}+\frac{\log(2/\alpha)}{2n}-\sqrt{\frac{(\log(2/\alpha))^2}{4n^2}+\frac{2\hat{p}\log(2/\alpha)}{n}}
\end{align*}
is a valid $(1-\alpha)$-level CI for $p$, for any finite sample $n$. That is, $\mathbb{P}(\hat{p}_l^{Chernoff}\leq p\leq \hat{p}_u^{Chernoff})\geq1-\alpha$ for any $n$.\label{bound con}
\end{theorem}

\paragraph{B-E CI.}

We develop another CI for $p$ by inverting the B-E theorem. Here, we assume that $p$ is known to satisfy $p<\frac12$ a priori (which is reasonable if we consider rare event).  In this paper, we use the standard version of the B-E theorem, and a potential future investigation is to consider a B-E bound for the studentized statistic \cite{wang1999exponential,wang2009relative}.

By the B-E theorem, we have that 
\begin{align*}
    \left|\mathbb{P}\left((\hat{p}-p)\sqrt{\frac{n}{p(1-p)}}\le x\right)-\Phi(x)\right|&\le \frac{C}{\sqrt{np(1-p)}},\\
    \left|\mathbb{P}\left((p-\hat p)\sqrt{\frac{n}{p(1-p)}}\le x\right)-\Phi(x)\right|&\le \frac{C}{\sqrt{np(1-p)}}
\end{align*}
where $C$ is a universal constant. We replace 
$x$ by $\frac{\hat{p}-p}{\sqrt{p(1-p)/n}}$ in the first inequality and $\frac{p-\hat{p}}{\sqrt{p(1-p)/n}}$ in the second one. Then we get that 
\begin{align*}
    \left|F(\hat{p})-\Phi\left(\frac{\hat{p}-p}{\sqrt{p(1-p)/n}}\right)\right|&\leq\frac{C}{\sqrt{np(1-p)}};\\
    \left|1-F_-(\hat{p})-\Phi\left(\frac{p-\hat{p}}{\sqrt{p(1-p)/n}}\right)\right|&\leq\frac{C}{\sqrt{np(1-p)}}.
\end{align*}
Hence $F(\hat{p})\geq\alpha/2,F_-(\hat{p})\leq1-\alpha/2$ implies that either 
\begin{equation*}
    p\geq\hat{p}\text{ and }\Phi\left(\frac{\hat{p}-p}{\sqrt{p(1-p)/n}}\right)+\frac{C}{\sqrt{np(1-p)}}\geq\alpha/2
\end{equation*}
or
\begin{equation*}
    p\leq\hat{p}\text{ and }\Phi\left(\frac{p-\hat{p}}{\sqrt{p(1-p)/n}}\right)+\frac{C}{\sqrt{np(1-p)}}\geq\alpha/2.
\end{equation*}
Thus, 
\begin{equation*}
    \left\{0<p\leq \hat p:\Phi\left(\frac{p-\hat{p}}{\sqrt{p(1-p)/n}}\right)+\frac{C}{\sqrt{np(1-p)}}\geq\frac{\alpha}{2}\right\}\cup\left\{\hat p\leq p<1:\Phi\left(\frac{\hat{p}-p}{\sqrt{p(1-p)/n}}\right)+\frac{C}{\sqrt{np(1-p)}}\geq\frac{\alpha}{2}\right\}
\end{equation*}
is a valid $(1-\alpha)$-level confidence region for $p$. Since we have assumed that $p<1/2$, the above confidence region can be further shrunk. To summarize, we have the following theorem:
\begin{theorem}[Validity of B-E CI under Standard Setting]
Assume that $p<1/2$. Then the set 
\begin{align}
    &\left\{0<p\leq \hat p\wedge\frac12:\Phi\left(\frac{p-\hat{p}}{\sqrt{p(1-p)/n}}\right)+\frac{C}{\sqrt{np(1-p)}}\geq\frac{\alpha}{2}\right\}\nonumber\\
    \cup&\left\{\hat p\leq p<\frac12:\Phi\left(\frac{\hat{p}-p}{\sqrt{p(1-p)/n}}\right)+\frac{C}{\sqrt{np(1-p)}}\geq\frac{\alpha}{2}\right\}
    \label{eqn:BE_region}
\end{align}
is a valid $(1-\alpha)$-level confidence region for $p$, for any finite sample $n$. 
\end{theorem}

\subsection{Analyses of Half-Widths}\label{sec:standard_width}

\paragraph{CLT CI.} Clearly, $\hat{p}_u^{CLT}-\hat p=\Theta(\sqrt{\hat p/n})=\Theta(\hat p/\sqrt{\hat{s}})$ and $\hat{p}-\hat{p}_l^{CLT}=\Theta(\hat p/\sqrt{\hat{s}})$. As explained in Section \ref{sec:problem}, we express the half-width as $\Theta(\hat p/\sqrt{\hat{s}})$ instead of $\Theta(\sqrt{\hat p/n})$ in order to understand the magnitude of the relative half-width with respect to $\hat p$ more clearly.

\paragraph{Wilson's CI.} 

We derive by some algebraic manipulations that

\begin{theorem}[Half-Width of Wilson's CI under Standard Setting]
\begin{align*}
|\hat{p}_u^{Wilson}-\hat{p}_u^{CLT}|&\leq \frac{z_{1-\alpha/2}^2}{n}+\frac{z_{1-\alpha/2}^3}{2n^{3/2}};\\
|\hat{p}_l^{Wilson}-\hat{p}_l^{CLT}|&\leq \frac{z_{1-\alpha/2}^2}{n}+\frac{z_{1-\alpha/2}^3}{2n^{3/2}}.
\end{align*}
\label{thm:Wilson_error}
\end{theorem}
Note that $1/n=\hat p/\hat{s}$, so the difference between the Wilson's CI and the CLT CI is of order $O(\hat p/\hat{s})$, which is of higher order than $\hat p/\sqrt{\hat{s}}$ in $\hat{s}$. In fact, as long as $\hat{s}\geq 1$, that is, we have at least one positive observation, then $\hat p/\sqrt{\hat{s}}=\sqrt{\hat{p}/n}=\sqrt{\hat{s}}/n\geq 1/n$. Since the half-width of the CLT CI is of order $\hat p/\sqrt{\hat{s}}$, we get that the half-width of the Wilson's CI is close to the CLT CI. 

\paragraph{Chernoff's CI.} 

When $\hat p=0$, Chernoff's CI reduces to $[0, 2\log(1/\alpha)/n]$ (and in fact we can construct even tighter bounds by using the binomial distribution of $n\hat p$ directly in this case). On the other hand, when $\hat p>0$, we can re-express using $\hat{s}=n\hat p$ to get 
\begin{align*}
    \hat{p}^{Chernoff}_u&=\hat{p}\left(1+\frac{\log(2/\alpha)}{\hat{s}}+\sqrt{\frac{(\log(2/\alpha))^2}{\hat{s}^2}+\frac{2\log(2/\alpha)}{\hat{s}}}\right),\\
    \hat{p}^{Chernoff}_l&=\hat{p}\left(1+\frac{\log(2/\alpha)}{2\hat{s}}-\sqrt{\frac{(\log(2/\alpha))^2}{4\hat{s}^2}+\frac{2\log(2/\alpha)}{\hat{s}}}\right).
\end{align*}
We highlight that in this case, the half-width of Chernoff's CI is of order $\Theta(\hat p/\sqrt{\hat{s}})$, which scales in the same order as the CLT CI. If we check the difference between this interval and the CLT interval, then we will find that it is of the same order as the half-width of the CLT CI. The following theorem presents the details of this claim. We will contrast this result with another one presented momentarily.
\begin{theorem}[Half-Width of Chernoff's CI under Standard Setting]
\begin{align*}
    \hat{p}_u^{Chernoff}-\hat{p}_u^{CLT}&\geq (\sqrt{2\log(2/\alpha)}-z_{1-\alpha/2})\sqrt{\frac{\hat{p}}{n}}+\frac{\log(2/\alpha)}{n},\\
    \hat{p}_l^{CLT}-\hat{p}_l^{Chernoff}&\geq(\sqrt{2\log(2/\alpha)}-z_{1-\alpha/2})\sqrt{\frac{\hat{p}}{n}}-\frac{\log(2/\alpha)}{2n}.
\end{align*}
Note that $\sqrt{2\log(2/\alpha)}-z_{1-\alpha/2}>0$ for $0<\alpha<1$.\label{thm:chernoff_error}
\end{theorem}

We recall that $1/n=\hat p/\hat{s}$ is of higher order than $\sqrt{\hat p/\hat{s}}=\hat p/\sqrt{\hat{s}}$. Provided that $\sqrt{2\log(2/\alpha)}-z_{1-\alpha/2}>0$,  $\hat{p}_u^{Chernoff}-\hat{p}_u^{CLT}$ (or, $\hat{p}_l^{Chernoff}-\hat{p}_l^{CLT}$) is of no higher order than $\hat{p}_u^{CLT}-\hat{p}$ (or, $\hat{p}_l^{CLT}-\hat{p}$).

\paragraph{B-E CI.}

We focus on the confidence upper bound as, unfortunately, we cannot derive a non-trivial confidence lower bound from \eqref{eqn:BE_region} since any $0<p<1/2$ such that $C/\sqrt{np(1-p)}\geq\alpha/2$ is contained in this confidence region. Now we further relax \eqref{eqn:BE_region} to develop a more explicit upper bound. In particular, \eqref{eqn:BE_region} could be relaxed to 
\begin{equation*}
    \left\{0<p<1/2:\Phi\left(\frac{\hat{p}-p}{\sqrt{p(1-p)/n}}\right)+\frac{C}{\sqrt{np(1-p)}}\geq\alpha/2\right\}.
\end{equation*}

In fact, for any $0\le \lambda\le 1-\frac{4C}{\sqrt{n}\alpha}$, we have that 
$$
0<p<1/2,\frac{C}{\sqrt{np(1-p)}}\ge(1-\lambda)\alpha/2\Rightarrow 
0<p\le \frac{1-\sqrt{1-\frac{16C^2}{n(1-\lambda)^2\alpha^2}}}{2}
$$
and 
$$
0<p<1/2,\Phi\left(\frac{\hat{p}-p}{\sqrt{p(1-p)/n}}\right)\ge\lambda\alpha/2\Rightarrow 0<p\le \frac{1+\frac{2n\hat{p}}{z_{\lambda\alpha/2}^2}+\sqrt{1+\frac{4n\hat{p}(1-\hat{p})}{z_{\lambda\alpha/2}^2}}}{2\left(1+\frac{n}{z_{\lambda\alpha/2}^2}\right)}.
$$
Therefore, we get that 
$$
0<p\leq\left(\frac{1-\sqrt{1-\frac{16C^2}{n(1-\lambda)^2\alpha^2}}}{2}\right)\lor\left(\frac{1+\frac{2n\hat{p}}{z_{\lambda\alpha/2}^2}+\sqrt{1+\frac{4n\hat{p}(1-\hat{p})}{z_{\lambda\alpha/2}^2}}}{2\left(1+\frac{n}{z_{\lambda\alpha/2}^2}\right)}\right)
$$
is a $(1-\alpha)$-level CI. For simplicity, we denote the two parts as $U_1$ and $U_2$ respectively. One may note that $\lambda$ is not necessarily deterministic. Instead, it can be dependent on the data as long as it stays within the interval $[0,1-\frac{4C}{\sqrt{n}\alpha}]$. In fact, we may choose $\lambda$ carefully such that $U_1\leq U_2$ is guaranteed for sufficiently large $n$. Specifically, the following theorem proposes another valid CI. 
\begin{theorem}[Relaxed B-E CI under Standard Setting]
Assume that $p<1/2$. Let 
$$
\lambda = 1-\frac{2\tilde{C}}{\sqrt{n}\alpha}
$$
where 
$$\tilde{C}=\left(\frac{C}{\sqrt{\hat{p}(1-\hat{p})}}\right)\wedge \left(\frac{u\sqrt{n}\alpha}{2}\right).$$
Here, $u<1$ is any constant such that $\frac{4C^2}{u^2\alpha^2}<z_{(1-u)\alpha/2}^2$. In the case that $\hat{p}=0\text{ or }1$, naturally we set $\tilde{C}=u\sqrt{n}\alpha/2$. Then there exists $N_0$, which does not depend on $p$ and $\hat{p}$, such that for any $n>N_0$,
$$
\hat{p}_u^{BE}=\frac{1+\frac{2n\hat{p}}{z_{\lambda\alpha/2}^2}+\sqrt{1+\frac{4n\hat{p}(1-\hat{p})}{z_{\lambda\alpha/2}^2}}}{2\left(1+\frac{n}{z_{\lambda\alpha/2}^2}\right)},\hat{p}_l^{BE}=0
$$
is a valid $(1-\alpha)$-level CI for $p$. In particular, $N_0$ can be chosen as 
$$
\left(\frac{4C}{u\alpha}\right)^2\lor \frac{12z_{(1-u)\alpha/2}^2C^2}{z_{(1-u)\alpha/2}^2u^2\alpha^2-4C^2}.
$$
\label{thm:BE_bound}
\end{theorem}

Actually, $\hat{p}_u^{BE}$ itself is a valid $(1-\alpha/2)$-level confidence upper bound for $p$, which is higher than the nominal level $1-\alpha$. The series of relaxations makes this CI more and more conservative, but we will show that the upper bound still has similar scale with $\hat{p}_u^{CLT}$ and $\hat{p}_u^{Wilson}$. Namely, we can derive that $|\hat{p}_u^{BE}-\hat{p}_u^{CLT}|$ is bounded by order $1/n$. In other words, though $\hat{p}_u^{CLT}$ has undesirable coverage probability in the rare-event setting, it is not ``too far'' from a valid upper bound. The following theorem states this result.
\begin{theorem}[Half-Width of B-E CI under Standard Setting]
Assume that $p<1/2$. $\hat{p}_u^{BE}$ is as defined in Theorem \ref{thm:BE_bound}. Then there is a constant $C_0$ which does not depend on $p$ and $\hat p$ such that 
$$
|\hat{p}_u^{BE}-\hat{p}_u^{CLT}|\leq C_0/n.
$$
\label{thm:BE_error}
\end{theorem}

Note that the bound in Theorem \ref{thm:BE_error} can be rephrased as $    |\hat{p}_u^{BE}-\hat{p}_u^{CLT}|\leq C_0\hat p/\hat{s}$. In other words, $\hat p_u^{BE}$ differ from $\hat p_u^{CLT}$ by a higher order than the half-width of the CLT CI in terms of $\hat{s}$, while all quantities scale with $\hat p$ in a similar manner. Compared to Theorem \ref{thm:chernoff_error}, we see in Theorem \ref{thm:BE_error} that $\hat p_u^{BE}$ is substantially tighter than $\hat p_u^{Chernoff}$ when $\hat{s}$ increases, although due to the implicit constant $C_0$ it may not be the case for small $\hat{s}$.

\section{Developments under Targeted Stopping}\label{sec:stopping}
We now present our results for the targeted stopping setting, following the roadmap for the standard setting presented earlier. Namely, we first present the construction of Chernoff's CI and the B-E CI (Section \ref{sec:stopping_develop}), followed by the analyses of half widths for all CIs (Section \ref{sec:stopping_width}).

\subsection{Derivation of New Confidence Intervals}\label{sec:stopping_develop}

To construct the new CIs under the targeted stopping setting, we again relax the confidence region \eqref{eqn:pvalueN} via Chernoff's inequality and the B-E theorem. Nevertheless, now we need to deal with the distribution of $N$ instead of $\hat p$ as in \eqref{eqn:pvalue}. Hence, as we will show below, the specific derivations of applying Chernoff's inequality and the B-E theorem differ from the standard setting.

\paragraph{Chernoff's CI.}

First, we propose Chernoff's CI similar to the one in the standard setting. By Markov's inequality, we get that 
\begin{equation*}
    \mathbb{P}(N\geq x)\leq e^{-tx}\mathbb{E}(e^{tN})=e^{-tx}\left(\frac{pe^t}{1-(1-p)e^t}\right)^{n_0},0<t<-\log(1-p).
\end{equation*}
Then for $x>n_0/p$,
\begin{equation*}
    \mathbb{P}(N\geq x)\leq \min_{0<t<-\log(1-p)}e^{-tx}\left(\frac{pe^t}{1-(1-p)e^t}\right)^{n_0}=\frac{(1-p)^{x-n_0}x^xp^{n_0}}{(x-n_0)^{x-n_0}n_0^{n_0}}.
\end{equation*}
Similarly, 
\begin{equation*}
    \mathbb{P}(N\leq x)\leq e^{tx}\mathbb{E}(e^{-tN})=e^{tx}\left(\frac{pe^{-t}}{1-(1-p)e^{-t}}\right)^{n_0},t>0
\end{equation*}
and thus for $0<x<n_0/p$,
\begin{equation*}
    \mathbb{P}(N\leq x)\leq\min_{t>0}e^{tx}\left(\frac{pe^{-t}}{1-(1-p)e^{-t}}\right)^{n_0}=\frac{(1-p)^{x-n_0}x^xp^{n_0}}{(x-n_0)^{x-n_0}n_0^{n_0}}.
\end{equation*}
Therefore, $F_N(N)\geq\alpha/2,F_{N-}(N)\leq1-\alpha/2$ implies that either 
\begin{equation*}
    N\geq n_0/p\text{ and }\frac{(1-p)^{N-n_0}N^Np^{n_0}}{(N-n_0)^{N-n_0}n_0^{n_0}}\geq \alpha/2
\end{equation*}
or
\begin{equation*}
   N\leq n_0/p\text{ and } \frac{(1-p)^{N-n_0}N^Np^{n_0}}{(N-n_0)^{N-n_0}n_0^{n_0}}\geq \alpha/2.
\end{equation*}
Finally we get that 
\begin{equation*}
    \left\{0<p<1:\frac{(1-p)^{N-n_0}N^Np^{n_0}}{(N-n_0)^{N-n_0}n_0^{n_0}}\geq \alpha/2\right\}
\end{equation*}
is a valid $(1-\alpha)$-level confidence region for $p$ under the targeted stopping setting. After simplification, we summarize our result with the following theorem:
\begin{theorem}[Validity of Chernoff's CI under Targeted Stopping]
Suppose that we keep sampling from $Bernoulli(p)$ until we get $n_0$ successes and the sample size is denoted by $N$. Then
\begin{equation}
    \left\{0<p<1:p^{n_0}(1-p)^{N-n_0}\geq \frac{\alpha}{2}\left(\frac{n_0}{N}\right)^{n_0}\left(1-\frac{n_0}{N}\right)^{N-n_0}\right\}
    \label{eqn:chernoffN}
\end{equation}
is a valid $(1-\alpha)$-level confidence region for $p$.
\end{theorem}

It is easy to verify that $f(p)=p^{n_0}(1-p)^{N-n_0}-\frac{\alpha}{2}(n_0/N)^{n_0}(1-n_0/N)^{N-n_0}$ is increasing in $[0,n_0/N]$ and decreasing in $[n_0/N,1]$. Moreover, we observe that $f(0)=f(1)<0$ and $f(n_0/N)>0$. Thus \eqref{eqn:chernoffN} is actually an interval and numerically, we could use the bisection method to compute the bounds. Nevertheless, this CI is not as easy to study analytically as under the standard setting, so we will not include its half-width result in the following section.

\paragraph{B-E CI.}

Now we apply the B-E theorem again. We still assume that $p<1/2$ is known as a priori. By the theorem, we get that
\begin{align}
    \left|\mathbb{P}\left(\frac{N-n_0/p}{\sqrt{n_0(1-p)/p^2}}\leq x\right)-\Phi(x)\right|&\leq \frac{C\rho_N}{\sigma_N^3\sqrt{n_0}},
    \label{eqn:BE_N1}\\
    \left|\mathbb{P}\left(\frac{n_0/p-N}{\sqrt{n_0(1-p)/p^2}}\leq x\right)-\Phi(x)\right|&\leq \frac{C\rho_N}{\sigma_N^3\sqrt{n_0}},
    \label{eqn:BE_N2}
\end{align}
where $\sigma_N^2=\mathbb{E}(N_i-1/p)^2=(1-p)/p^2$ and $\rho_N=E|N_i-1/p|^3$.

We need to deal with $\rho_N$ first. In fact, we know that 
$$
p^3\rho_N=p^3E\left|N_i-\frac{1}{p}\right|^3=E|pN_i-1|^3\le 1+3p\mathbb{E}(N_i)+3p^2\mathbb{E}(N_i^2)+p^3\mathbb{E}(N_i^3).
$$
Since $N_i\sim Geometric(p)$, we know that 
$$
\mathbb{E}(N_i)=\frac{1}{p},\mathbb{E}(N_i^2)=\frac{2-p}{p^2},\mathbb{E}(N_i^3)=\frac{p^2-6p+6}{p^3},
$$
and thus 
$$
p^3\rho_N\le p^2-6p+6+3(2-p)+3+1=p^2-12p+16\le 16.
$$
Hence, 
$$
\frac{C\rho_N}{\sigma_N^3\sqrt{n_0}}=\frac{Cp^3\rho_N}{(1-p)^{\frac32}\sqrt{n_0}}\le \frac{C'}{(1-p)^{\frac32}\sqrt{n_0}}
$$
where $C'=16C$ is an absolute constant.

By setting $x=\sqrt{\frac{p^2}{n_0(1-p)}}\left(N-\frac{n_0}{p}\right)$ in \eqref{eqn:BE_N1} and $x=\sqrt{\frac{p^2}{n_0(1-p)}}\left(\frac{n_0}{p}-N\right)$ in \eqref{eqn:BE_N2}, we get that 
\begin{align*}
    \left|F_N(N)-\Phi\left(\frac{Np-n_0}{\sqrt{n_0(1-p)}}\right)\right|&\leq \frac{C'}{\sqrt{n_0(1-p)^3}};\\
    \left|1-F_{N-}(N)-\Phi\left(\frac{n_0-Np}{\sqrt{n_0(1-p)}}\right)\right|&\leq \frac{C'}{\sqrt{n_0(1-p)^3}}.
\end{align*}
Hence $F_N(N)\geq\alpha/2,F_{N-}(N)\leq 1-\alpha/2$ implies that either
\begin{equation*}
    p\geq n_0/N\text{ and }\Phi\left(\frac{n_0-Np}{\sqrt{n_0(1-p)}}\right)+\frac{C'}{\sqrt{n_0(1-p)^3}}\geq\alpha/2
\end{equation*}
or 
\begin{equation*}
    p\leq n_0/N\text{ and }\Phi\left(\frac{Np-n_0}{\sqrt{n_0(1-p)}}\right)+\frac{C'}{\sqrt{n_0(1-p)^3}}\geq\alpha/2.
\end{equation*}
Thus, we develop a valid confidence region under this particular setting, which is similar to the one in Section \ref{sec:standard}:
\begin{theorem}[Validity of B-E CI under Targeted Stopping]
Suppose that we keep sampling from $Bernoulli(p)$ until we get $n_0$ successes and the sample size is denoted by $N$. Assume that $p<1/2$. Then 
\begin{align}
    &\left\{0< p\leq\frac{n_0}{N}\wedge\frac12:\Phi\left(\frac{Np-n_0}{\sqrt{n_0(1-p)}}\right)+\frac{C'}{\sqrt{n_0(1-p)^3}}\geq\frac{\alpha}{2}\right\}\nonumber\\
    \cup&\left\{\frac{n_0}{N}\leq p<\frac{1}{2}:\Phi\left(\frac{n_0-Np}{\sqrt{n_0(1-p)}}\right)+\frac{C'}{\sqrt{n_0(1-p)^3}}\geq\frac{\alpha}{2}\right\}
\label{eqn:BE_regionN}
\end{align}
is a valid $(1-\alpha)$-level confidence region for $p$. Here, $C'$ is a universal constant. In particular, one may pick $C'=16C$ where $C$ is the constant in the B-E theorem.
\label{thm:BE_regionN}
\end{theorem}

Note that when $n_0$ is not large enough, we have $C'/\sqrt{n_0(1-p)^3}\geq\alpha/2$ anyway. That is to say, this confidence region is not really practical. However, it could still provide an insight on how close are the CLT or Wilson's intervals to a valid one.


\subsection{Analyses of Half-Widths}\label{sec:stopping_width}

As mentioned before, under the targeted stopping setting, Chernoff's CI is no longer easy to analyze. Thus, in this subsection, we will only cover the analyses for the CLT, Wilson's and B-E CIs. For the first two, as the formulas of the CIs are the same as the standard setting, we simply present the results here. For the B-E CI, the main idea of the derivations is similar to the standard setting, but there are differences in the technical details. In particular, now we are able to get a non-trivial lower bound.

\paragraph{CLT CI and Wilson's CI.}

Under the targeted stopping setting, the formulas of the CLT CI and Wilson's CI are the same as under the standard setting with $\hat p=n_0/N$. Thus, clearly we still have that $\hat{p}_{u,n_0}^{CLT}-\hat{p}=\hat{p}-\hat{p}_{l,n_0}^{CLT}=\Theta(\hat{p}/\sqrt{n_0})=\Theta(\sqrt{n_0}/N)$ and that $|\hat{p}_{u,n_0}^{Wilson}-\hat{p}_{u,n_0}^{CLT}|=O(1/N)=O(\hat p/n_0),|\hat{p}_{l,n_0}^{Wilson}-\hat{p}_{l,n_0}^{CLT}|=O(1/N)=O(\hat p/n_0)$. 

\paragraph{B-E CI.}

Similar to Section \ref{sec:standard}, the confidence region \eqref{eqn:BE_regionN} could be further relaxed. However, unlike in the standard setting, now the error term $C'/\sqrt{n_0(1-p)^3}$ could be well controlled for tiny $p$ and as a result, we are able to get a non-trivial lower bound in this case. More concretely, for any $0<\lambda<1$, we have that 
\begin{equation*}
    0<p<1,\frac{C'}{\sqrt{n_0(1-p)^3}}\geq(1-\lambda)\alpha/2\Rightarrow p\geq 1-\left(\frac{4C'}{n_0(1-\lambda)^2\alpha^2}\right)^{1/3}.
\end{equation*}
If we could find $0<\lambda<1$ such that 
\begin{equation*}
    \left(\frac{4C'^2}{n_0(1-\lambda)^2\alpha^2}\right)^{1/3}=\frac12,
\end{equation*}
then for any $0<p<1/2$, we have that 
\begin{equation*}
    \frac{C'}{\sqrt{n_0(1-p)^3}}\leq(1-\lambda)\alpha/2.
\end{equation*}
As a result, any $p$ in \eqref{eqn:BE_regionN} must satisfy that 
\begin{equation*}
    0<p\leq \frac{n_0}{N}\wedge\frac12,\Phi\left(\frac{Np-n_0}{\sqrt{n_0(1-p)}}\right)\geq\lambda\alpha/2 \text{ or }\frac{n_0}{N}\leq p<\frac{1}{2},\Phi\left(\frac{n_0-Np}{\sqrt{n_0(1-p)}}\right)\geq\lambda\alpha/2.
\end{equation*}
After simplification, we get that
\begin{align*}
    0<p\leq \frac{n_0}{N},\Phi\left(\frac{Np-n_0}{\sqrt{n_0(1-p)}}\right)\geq\lambda\alpha/2&\Rightarrow p\geq\frac{2Nn_0-z_{\lambda\alpha/2}^2n_0-\sqrt{4z_{\lambda\alpha/2}^2Nn_0(N-n_0)+z_{\lambda\alpha/2}^4n_0^2}}{2N^2},\\
    \frac{n_0}{N}\leq p<1,\Phi\left(\frac{n_0-Np}{\sqrt{n_0(1-p)}}\right)\geq\lambda\alpha/2&\Rightarrow p\leq \frac{2Nn_0-z_{\lambda\alpha/2}^2n_0+\sqrt{4z_{\lambda\alpha/2}^2Nn_0(N-n_0)+z_{\lambda\alpha/2}^4n_0^2}}{2N^2}.
\end{align*}
Thus \eqref{eqn:BE_regionN} could be relaxed into a valid $(1-\alpha)$-level CI, which is defined more rigorously in the following theorem:
\begin{theorem}[Relaxed B-E CI under Targeted Stopping]
Suppose that we keep sampling from $Bernoulli(p)$ until we get $n_0$ successes and the sample size is denoted by $N$. Assume that $p<1/2$. Let 
\begin{equation*}
    \lambda=1-\frac{4\sqrt{2}C'}{\sqrt{n_0}\alpha}.
\end{equation*}
Then for any $n_0>32C'^2/\alpha^2$, we have that 
\begin{align*}
    \hat{p}_{u,n_0}^{BE}&=\frac{2Nn_0-z_{\lambda\alpha/2}^2n_0+\sqrt{4z_{\lambda\alpha/2}^2Nn_0(N-n_0)+z_{\lambda\alpha/2}^4n_0^2}}{2N^2},\\
    \hat{p}_{l,n_0}^{BE}&=\frac{2Nn_0-z_{\lambda\alpha/2}^2n_0-\sqrt{4z_{\lambda\alpha/2}^2Nn_0(N-n_0)+z_{\lambda\alpha/2}^4n_0^2}}{2N^2}
\end{align*}
is a valid $(1-\alpha)$-level CI for $p$. Here, $C'$ is the same as in Theorem \ref{thm:BE_regionN}.
\label{thm:BE_boundN}
\end{theorem}

Finally, like in the standard setting, we will compare the difference between the B-E CI and the CLT CI.
\begin{theorem}[Half-Width of B-E CI under Targeted Stopping]
Assume that $p<1/2$. $\hat{p}_{u,n_0}^{BE}$ and $\hat{p}_{l,n_0}^{BE}$ are as defined in Theorem \ref{thm:BE_boundN}. Then there is a constant $C_0'$ which does not depend on $p$ and $N$ such that 
\begin{align*}
    \hat{p}_{u,n_0}^{BE}-\hat{p}_{u,n_0}^{CLT}&\leq C_0'/N,\\
    \hat{p}_{l,n_0}^{CLT}-\hat{p}_{l,n_0}^{BE}&\leq C_0'/N.
\end{align*}
\label{thm:BE_errorN}
\end{theorem}

Therefore, under the targeted stopping setting, we could justify that the CLT CI is not too far from a valid one in terms of both upper bound and lower bound.

\section{Numerical Experiments}
\label{sec:numerical}
To close this paper, we perform some numerical experiments to visualize the differences among the CIs.
\subsection{Experiments under Standard Setting}
The true value is chosen as $p=1e-6$. For each of the settings $n=5/p, 10/p, 30/p, 50/p, 100/p$, we conduct 100,000 experimental repetitions and calculate the CIs with $\alpha=0.05$. Figure \ref{fig:standard} and Table \ref{tab:standard} respectively present 
the average values of the confidence upper and lower bounds from the 1000 repetitions
and the coverage probabilities of the five CIs covered in this paper. 
\begin{table}[h]
    \begin{minipage}{0.45\linewidth}
		\centering
		\includegraphics[width=\textwidth]{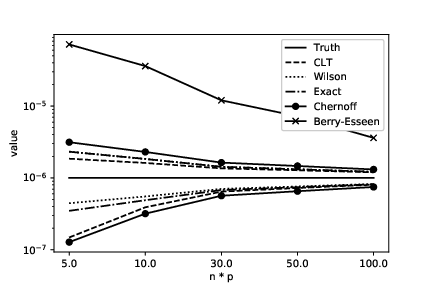}
		\captionof{figure}{Average values of the confidence upper and lower bounds under the standard setting.}
		\label{fig:standard}
	\end{minipage}
	\begin{minipage}{0.5\linewidth}
		\centering
		\begin{tabular}{rrrrrrrrr}
            \multicolumn{1}{l}{$np$}
          & \multicolumn{1}{l}{CLT} & \multicolumn{1}{l}{Wilson} & \multicolumn{1}{l}{Exact} & \multicolumn{1}{l}{Chernoff} & \multicolumn{1}{l}{BE} \\
    5     & 0.86990      &0.96252     &0.97992        &0.99930      &1.0   \\
    10    & 0.92601      &0.96252     &0.97542        &0.99767      &1.0 \\
    30    & 0.93100      &0.94476     &0.95531        &0.99560      &1.0  \\
    50    & 0.95098      &0.94596     &0.95451        &0.99659      &1.0   \\
    100   & 0.94501      &0.94920     &0.95485        &0.99575      &1.0  \\
    \end{tabular}
    \caption{Coverage probabilities of the CIs under the standard setting.}
    \label{tab:standard}
	\end{minipage}\hfill
\end{table}

As analyzed in Section \ref{sec:standard_width}, when $np$ is large, the CIs scale similarly, except that B-E fails to give a non-zero lower bound. While the CLT interval is closest to the truth in terms of the mean value of the upper bound, it is not reliable especially when $np$ is small. For instance, when $np=5$, its coverage probability is only $0.858$, which is much lower than the nominal confidence level $0.95$. Wilson's and the exact CIs are quite similar, especially for the upper bound. However, we notice that Wilson's bound sometimes fails to achieve the nominal confidence level, but the error in the coverage probability is acceptable to some extent. Chernoff's and the B-E CIs are conservative as expected, since they further relax the conservative confidence region \eqref{eqn:pvalue}. We would like to point out that though the B-E upper bound seems to be much larger than Chernoff's, it decays much faster as $np$ increases, which coincides with Theorems \ref{thm:chernoff_error} and \ref{thm:BE_error}.  

\subsection{Experiments under Targeted Stopping}
Now we consider the targeted stopping setting. We still set $p=1e-6$. For each of the settings $n_0=5,10,30,50,100$, we conduct 1000 experimental repetitions and calculate the CIs with $\alpha=0.05$. Figure \ref{fig:targeted} and Table \ref{tab:targeted} respectively present the average values of the confidence upper and lower bounds from the 100,000 repetitions and the coverage probabilities of the CIs. Note that we do not include the B-E CI since it is trivial due to small $n_0$ as aforementioned.
\begin{table}[h]
    \begin{minipage}{0.45\linewidth}
		\centering
		\includegraphics[width=\textwidth]{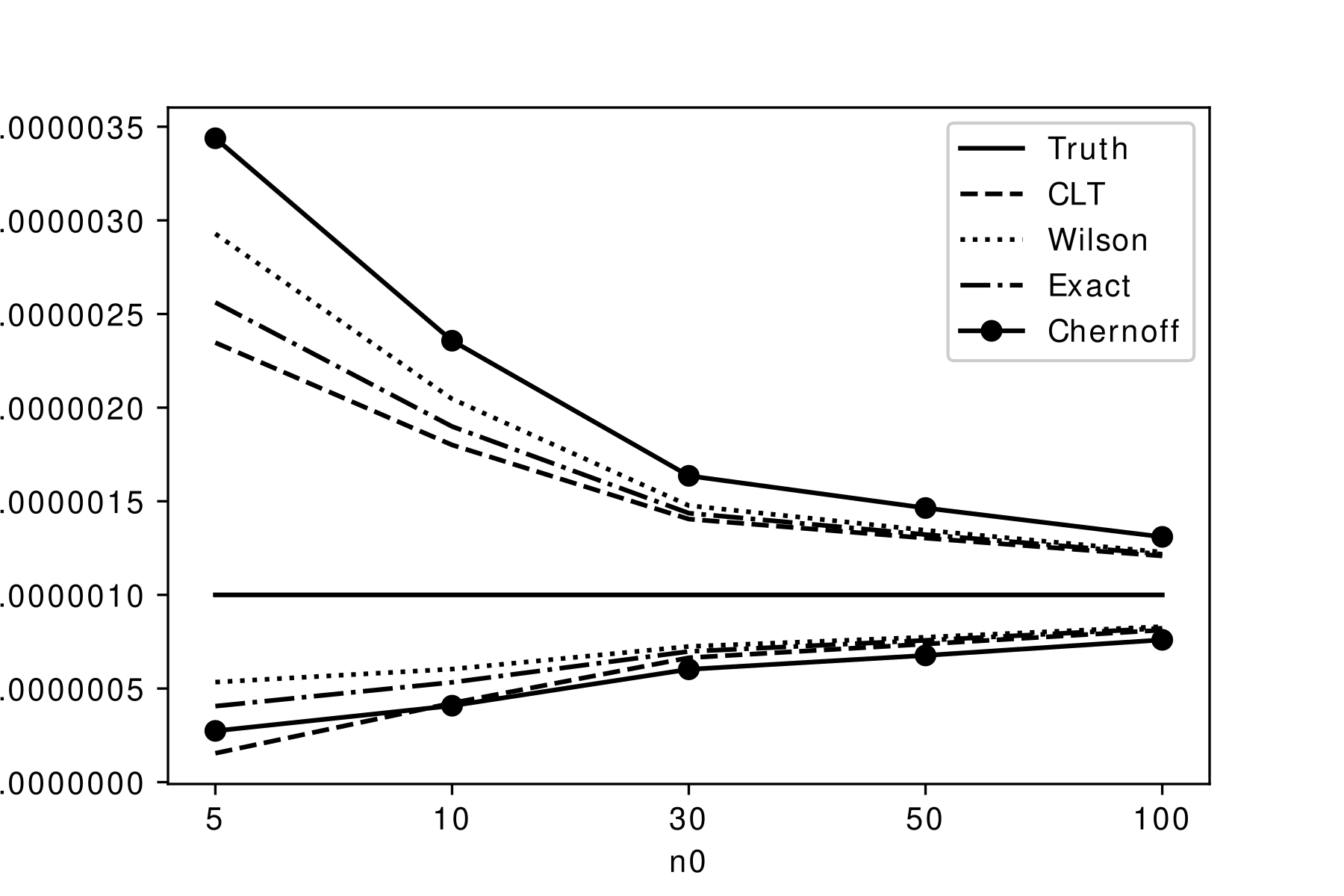}
		\captionof{figure}{Average values of the confidence upper and lower bounds under the targeted stopping setting.}
		\label{fig:targeted}
	\end{minipage}
	\begin{minipage}{0.5\linewidth}
		\centering
		\begin{tabular}{rrrrrrrrr}
            \multicolumn{1}{l}{$n_0$}
          & \multicolumn{1}{l}{CLT} & \multicolumn{1}{l}{Wilson} & \multicolumn{1}{l}{Exact} & \multicolumn{1}{l}{Chernoff}  \\
    5     & 0.95660      &0.92470     &0.94998        &0.99233   \\
    10    & 0.95511      &0.93714     &0.95015        &0.99300 \\
    30    & 0.95120      &0.94564     &0.94951        &0.99328  \\
    50    & 0.95169      &0.94750     &0.95025        &0.99276  \\
    100   & 0.95077      &0.94890     &0.95007        &0.99337  \\
    \end{tabular}
    \caption{Coverage probabilities of the CIs under the targeted stopping setting.}
    \label{tab:targeted}
	\end{minipage}\hfill
\end{table}

Chernoff's CI is still conservative as expected, since it further relaxes the conservative confidence region \eqref{eqn:pvalueN}. We focus on comparing the other three CIs. From Figure \ref{fig:targeted}, we see that for the upper bound, the CLT bound is the closest to the truth and Wilson's bound is the farthest. For the lower bound, Wilson's bound is the closest and the CLT bound is the farthest, The exact bound falls in between. On the other hand, they all have a similar magnitude, and when $n_0$ is large, they are all close to each other. In terms of coverage probability (Table \ref{tab:targeted}), the three CIs also have similar performances when $n_0$ is large. In particular, in relation to the motivation in Section \ref{sec:problem}, if we can sample until we observe enough (say 30) successes, then the CLT CI is indeed reliable to use although it is not guaranteed to be valid.

\section{Conclusion}\label{sec:conclusion} 
In this paper, we study the uncertainty quantification, more precisely the construction, validity and tightness of CIs, for rare-event probability using naive sample proportion estimators from Bernoulli data. We focus on two settings, the standard setting where the sample size is fixed, and the targeted stopping setting where the number of successes is fixed. Under each setting, we first review the existing CLT, Wilson's and Exact CIs. It is known that the CLT and Wilson's CIs are not necessarily valid in the sense that the coverage probability can be lower than the nominal confidence level, and the exact CI is valid yet its tightness is hard to analyze. These motivate us to derive other valid CIs with more explicit expressions. More specifically, we relax the exact confidence region via inverting Chernoff's inequality and the B-E theorem to obtain Chernoff's and the B-E CIs respectively. Table \ref{tab:summary} in Section \ref{sec:problem_conclusion} is a comprehensive summary of our findings, and we briefly summarize our key findings in Section \ref{sec:problem_conclusion}.


Overall, we recommend the Exact CI when one wants to ensure a guarantee of the nominal confidence level, otherwise we suggest using Wilson's CI given its excellent empirical performance. Moreover, either in the standard setting or the targeted stopping setting, we have justified that the CLT CI is not far from the valid Chernoff's and B-E CIs, so it also exhibits reasonable tightness in terms of half-width. However, its coverage probability can deviate significantly from the nominal level when $np$ is small. The latter two intervals are conservative and hence not recommended to use in practice, but they provide us useful insights in understanding that the CLT and Wilson's CIs are relatively close to these two valid CIs.

\appendix

\section{Proofs}
\begin{proof}[Proof of Theorem \ref{thm:Wilson_error}]
We get directly from the formula of $\hat{p}_u^{Wilson}$ that 
$$
\begin{aligned}
|\hat{p}_u^{Wilson}-\hat{p}_u^{CLT}|&=\left|\frac{1-2\hat{p}+\sqrt{1+\frac{4n\hat{p}(1-\hat{p})}{z_{1-\alpha/2}^2}}-2\left(1+\frac{n}{z_{1-\alpha/2}^2}\right)z_{1-\alpha/2}\sqrt{\frac{\hat{p}(1-\hat{p})}{n}}}{2\left(1+\frac{n}{z_{1-\alpha/2}^2}\right)}\right|\\
&\leq \frac{|1-2\hat{p}|+\left|2z_{1-\alpha/2}\sqrt{\frac{\hat{p}(1-\hat{p})}{n}}\right|+\left|\sqrt{1+\frac{4n\hat{p}(1-\hat{p})}{z_{1-\alpha/2}^2}}-\frac{2n}{z_{1-\alpha/2}}\sqrt{\frac{\hat{p}(1-\hat{p})}{n}}\right|}{2n/z_{1-\alpha/2}^2}.
\end{aligned}
$$
We have that $|1-2\hat{p}|\leq 1$, $\left|2z_{1-\alpha/2}\sqrt{\hat{p}(1-\hat{p})/n}\right|\leq z_{1-\alpha/2}/\sqrt{n}$ and 
$$
\left|\sqrt{1+\frac{4n\hat{p}(1-\hat{p})}{z_{1-\alpha/2}^2}}-\frac{2n}{z_{1-\alpha/2}}\sqrt{\frac{\hat{p}(1-\hat{p})}{n}}\right|=\left|\frac{1}{\sqrt{1+\frac{4n\hat{p}(1-\hat{p})}{z_{1-\alpha/2}^2}}+\frac{2n}{z_{1-\alpha/2}}\sqrt{\frac{\hat{p}(1-\hat{p})}{n}}}\right|\leq 1,
$$
which concludes the proof for the upper bounds. The proof for the lower bound is almost the same.
\end{proof}
\begin{proof}[Proof of Theorem \ref{thm:chernoff_error}]
We have that 
$$
\begin{aligned}
\hat{p}_u^{Chernoff}-\hat{p}_u^{CLT}&=\sqrt{\frac{2\log(2/\alpha)\hat{p}}{n}+\frac{(\log(2/\alpha))^2}{n^2}}+\frac{\log(2/\alpha)}{n}-z_{1-\alpha/2}\sqrt{\frac{\hat{p}(1-\hat{p})}{n}}\\
&\geq (\sqrt{2\log(2/\alpha)}-z_{1-\alpha/2})\sqrt{\frac{\hat{p}}{n}}+\frac{\log(2/\alpha)}{n}.
\end{aligned}
$$
Similarly we could prove the inequality for the lower bounds.
\end{proof}
\begin{proof}[Proof of Theorem \ref{thm:BE_bound}]
Following our derivations, it suffices to show that for the given $N_0$ and any $n>N_0$, we have that $0\leq \lambda\leq 1-\frac{4C}{\sqrt{n}\alpha}$ and $U_1\leq U_2$. Obviously, $\frac{2\tilde{C}}{\sqrt{n}\alpha}\leq u<1$, so $\lambda>0$. On the other side, $\lambda\leq 1-\frac{4C}{\sqrt{n}\alpha}$ holds since $n\geq \left(\frac{4C}{u\alpha}\right)^2$. Now we prove that $U_1\leq U_2$ for $n>N_0$. Indeed, if $\tilde{C}=C/\sqrt{\hat{p}(1-\hat{p})}$, then $U_1=\hat{p}\land(1-\hat{p})\le \hat{p}$ and we know that $U_2\ge\hat{p}$, so $U_1\le U_2$. In the other case that $\tilde{C}=u\sqrt{n}\alpha/2$, we have that
$$
U_1=\frac{1-\sqrt{1-\frac{16C^2}{nu^2\alpha^2}}}{2}=\frac{\frac{16C^2}{nu^2\alpha^2}}{2\left(1+\sqrt{1-\frac{16C^2}{nu^2\alpha^2}}\right)}\leq\frac{\frac{16C^2}{nu^2\alpha^2}}{2\left(1+1-\frac{16C^2}{nu^2\alpha^2}\right)}=\frac{4C^2}{nu^2\alpha^2-8C^2}
$$
and 
$$
U_2\ge \frac{1}{1+\frac{n}{z_{(1-u)\alpha/2}^2}}=\frac{z_{(1-u)\alpha/2}^2}{z_{(1-u)\alpha/2}^2+n}.
$$
Since $u$ is chosen such that $\frac{4C^2}{u^2\alpha^2}<z_{(1-u)\alpha/2}^2$ and 
$$
n\geq \frac{12z_{(1-u)\alpha/2}^2C^2}{z_{(1-u)\alpha/2}^2u^2\alpha^2-4C^2},
$$
we get that 
$$
\frac{4C^2}{nu^2\alpha^2-8C^2}\leq \frac{z_{(1-u)\alpha/2}^2}{z_{(1-u)\alpha/2}^2+n}
$$
and hence $U_1\leq U_2$. Note that as $u\uparrow 1$, $\frac{4C^2}{u^2\alpha^2}\rightarrow \frac{4C^2}{\alpha^2}$ while $z_{(1-u)\alpha/2}^2\rightarrow\infty$, and thus such $u$ exists.
\end{proof}
\begin{proof}[Proof of Theorem \ref{thm:BE_error}]
We have that 
$$
\hat{p}_u^{BE}-\hat{p}_u^{CLT}=\frac{1-2\hat{p}+\sqrt{1+\frac{4n\hat{p}(1-\hat{p})}{z_{\lambda\alpha/2}^2}}-2\left(1+\frac{n}{z_{\lambda\alpha/2}^2}\right)z_{1-\alpha/2}\sqrt{\frac{\hat{p}(1-\hat{p})}{n}}}{2\left(1+\frac{n}{z_{\lambda\alpha/2}^2}\right)}.
$$
We first deal with 
$$
\begin{aligned}
&\sqrt{1+\frac{4n\hat{p}(1-\hat{p})}{z_{\lambda\alpha/2}^2}}-2\left(1+\frac{n}{z_{\lambda\alpha/2}^2}\right)z_{1-\alpha/2}\sqrt{\frac{\hat{p}(1-\hat{p})}{n}}\\
=&\frac{\left(1-\frac{8z_{1-\alpha/2}^2\hat{p}(1-\hat{p})}{z_{\lambda\alpha/2}^2}\right)-\frac{4z_{1-\alpha/2}^2\hat{p}(1-\hat{p})}{n}+\frac{4n\hat{p}(1-\hat{p})}{z_{\lambda\alpha/2}^2}\left(1-\frac{z_{1-\alpha/2}^2}{z_{\lambda\alpha/2}^2}\right)}{\sqrt{1+\frac{4n\hat{p}(1-\hat{p})}{z_{\lambda\alpha/2}^2}}+2\left(1+\frac{n}{z_{\lambda\alpha/2}^2}\right)z_{1-\alpha/2}\sqrt{\frac{\hat{p}(1-\hat{p})}{n}}}.
\end{aligned}
$$
The denominator satisfies that 
$$
\sqrt{1+\frac{4n\hat{p}(1-\hat{p})}{z_{\lambda\alpha/2}^2}}+2\left(1+\frac{n}{z_{\lambda\alpha/2}^2}\right)z_{1-\alpha/2}\sqrt{\frac{\hat{p}(1-\hat{p})}{n}}\ge\left( \frac{2\sqrt{n\hat{p}(1-\hat{p})}(z_{1-\lambda\alpha/2}+z_{1-\alpha/2})}{z_{\lambda\alpha/2}^2}\right)\lor 1.	
$$
Note that $(z_{1-\lambda\alpha/2}+z_{1-\alpha/2})/z_{\lambda\alpha/2}^2$ increases with the value of $\lambda$. Since $\lambda\ge 1-u>0$, we can find a constant $C_1$ such that 
$$
\sqrt{1+\frac{4n\hat{p}(1-\hat{p})}{z_{\lambda\alpha/2}^2}}+2\left(1+\frac{n}{z_{\lambda\alpha/2}^2}\right)z_{1-\alpha/2}\sqrt{\frac{\hat{p}(1-\hat{p})}{n}}\ge \left(C_1\sqrt{n\hat{p}(1-\hat{p})}\right)\lor 1.
$$
Then we deal with the numerator. We know that $z_{\alpha/2}=\Phi^{-1}(\alpha/2)$ and $z_{\lambda\alpha/2}=\Phi^{-1}(\lambda\alpha/2)$. By Taylor expansion, we have that 
$$
\frac{1}{z_{\lambda\alpha/2}^2}=\frac{1}{z_{\alpha/2}^2}-\frac{2\sqrt{2\pi}}{z_{\alpha/2}^3}e^{\frac{z_{\alpha/2}^2}{2}}(\lambda-1)\alpha/2+r(\lambda).
$$
Here, $r(\lambda)$ is continuous in $\lambda$ and $r(\lambda)/(1-\lambda)\to 0$ as $\lambda\uparrow 1$. We also know that $1-\lambda\le u$, and thus $|r(\lambda)/(1-\lambda)|=|(\sqrt{n}\alpha r(\lambda))/(2\tilde{C})|$ is bounded by a constant. Hence $|\sqrt{n\hat{p}(1-\hat{p})}r(\lambda)|$ is bounded by a constant. We have that
$$
1-\frac{z_{1-\alpha/2}^2}{z_{\lambda\alpha/2}^2}=\frac{2\sqrt{2\pi}}{z_{1-\alpha/2}}e^{\frac{z_{1-\alpha/2}^2}{2}}\frac{\tilde{C}}{\sqrt{n}}-z_{1-\alpha/2}^2r(\lambda).
$$
Thus, the numerator satisfies that 
$$
\begin{aligned}
& \left|\left(1-\frac{8z_{1-\alpha/2}^2\hat{p}(1-\hat{p})}{z_{\lambda\alpha/2}^2}\right)-\frac{4z_{1-\alpha/2}^2\hat{p}(1-\hat{p})}{n}+\frac{4n\hat{p}(1-\hat{p})}{z_{\lambda\alpha/2}^2}\left(1-\frac{z_{1-\alpha/2}^2}{z_{\lambda\alpha/2}^2}\right)\right|\\
\le & 1+8\hat{p}(1-\hat{p})+\frac{4z_{1-\alpha/2}^2\hat{p}(1-\hat{p})}{n}+\frac{4n\hat{p}(1-\hat{p})}{z_{\lambda\alpha/2}^2}\left(\frac{2\sqrt{2\pi}}{z_{1-\alpha/2}}e^{\frac{z_{1-\alpha/2}^2}{2}}\frac{\tilde{C}}{\sqrt{n}}-z_{1-\alpha/2}^2r(\lambda)\right).
\end{aligned}
$$
Clearly, the first three terms divided by the denominator are bounded by some constants. Now we consider the fourth term. Since $|\sqrt{n\hat{p}(1-\hat{p})}r(\lambda)|$ is bounded, we can also get that the fourth term divided by the denominator is bounded by some universal constant.
\par 
Therefore, combining the above results, we know that 
$$
\left|1-2\hat{p}+\sqrt{1+\frac{4n\hat{p}(1-\hat{p})}{z_{\lambda\alpha/2}^2}}-2\left(1+\frac{n}{z_{\lambda\alpha/2}^2}\right)z_{1-\alpha/2}\sqrt{\frac{\hat{p}(1-\hat{p})}{n}}\right|\le C_2
$$
where $C_2$ is a positive constant. We also have that 
$$
2\left(1+\frac{n}{z_{\lambda\alpha/2}^2}\right)\ge \frac{2n}{z_{(1-u)\alpha/2}^2}.
$$
Hence the error term satisfies that
$$
\left|\frac{1-2\hat{p}+\sqrt{1+\frac{4n\hat{p}(1-\hat{p})}{z_{\lambda\alpha}^2}}-2\left(1+\frac{n}{z_{\lambda\alpha}^2}\right)z_{1-\alpha}\sqrt{\frac{\hat{p}(1-\hat{p})}{n}}}{2\left(1+\frac{n}{z_{\lambda\alpha}^2}\right)}\right|\le \frac{C_0}{n}
$$
for some constant $C_0$. From the above derivations, we find that $C_0$ only depends on $\alpha$ and the choice of $u$. 
\end{proof}
\begin{proof}[Proof of Theorem \ref{thm:BE_boundN}]
If $n_0>32C'^2/\alpha^2$, then we get that $0<\lambda<1$. The validness of the CI is justified by the derivations above the theorem.
\end{proof}
\begin{proof}[Proof of Theorem \ref{thm:BE_errorN}]
We have that 
\begin{align*}
    \hat{p}_{u,n_0}^{BE}-\hat{p}_{u,n_0}^{CLT}&=-\frac{z_{\lambda\alpha/2}^2n_0}{2N^2}+\sqrt{\frac{z_{\lambda\alpha/2}^2n_0(N-n_0)}{N^3}+\frac{z_{\lambda\alpha/2}^4n_0^2}{4N^4}}-z_{1-\alpha/2}\sqrt{\frac{n_0(N-n_0)}{N^3}}\\
    &=\frac{\frac{(z_{\lambda\alpha/2}^2-z_{1-\alpha/2}^2)n_0(N-n_0)}{N^3}-\frac{n_0z_{\lambda\alpha/2}^2z_{1-\alpha/2}}{N^2}\sqrt{\frac{n_0(N-n_0)}{N^3}}}{\sqrt{\frac{z_{\lambda\alpha/2}^2n_0(N-n_0)}{N^3}+\frac{z_{\lambda\alpha/2}^4n_0^2}{4N^4}}+\frac{z_{\lambda\alpha/2}^2n_0}{2N^2}+z_{1-\alpha/2}\sqrt{\frac{n_0(N-n_0)}{N^3}}}
\end{align*}
First, we have that the denominator
\begin{equation*}
    \sqrt{\frac{z_{\lambda\alpha/2}^2n_0(N-n_0)}{N^3}+\frac{z_{\lambda\alpha/2}^4n_0^2}{4N^4}}+\frac{z_{\lambda\alpha/2}^2n_0}{2N^2}+z_{1-\alpha/2}\sqrt{\frac{n_0(N-n_0)}{N^3}}\geq \sqrt{\frac{z_{\lambda\alpha/2}^2n_0(N-n_0)}{N^3}}=z_{1-\lambda\alpha/2}\sqrt{\frac{n_0(N-n_0)}{N^3}}.
\end{equation*}
Next, we have that the numerator 
\begin{equation*}
    \frac{(z_{\lambda\alpha/2}^2-z_{1-\alpha/2}^2)n_0(N-n_0)}{N^3}-\frac{n_0z_{\lambda\alpha/2}^2z_{1-\alpha/2}}{N^2}\sqrt{\frac{n_0(N-n_0)}{N^3}}\leq \left(1-\frac{z_{1-\alpha/2}^2}{z_{\lambda\alpha/2}^2}\right)\frac{z_{\lambda\alpha/2}^2n_0(N-n_0)}{N^3}.
\end{equation*}
As mentioned in the proof of Theorem \ref{thm:BE_error}, we have that 
$$
\frac{1}{z_{\lambda\alpha/2}^2}=\frac{1}{z_{\alpha/2}^2}-\frac{2\sqrt{2\pi}}{z_{\alpha/2}^3}e^{\frac{z_{\alpha/2}^2}{2}}(\lambda-1)\alpha/2+r(\lambda)
$$
and 
$$
1-\frac{z_{1-\alpha/2}^2}{z_{\lambda\alpha/2}^2}=\frac{\sqrt{2\pi}}{z_{1-\alpha/2}}e^{\frac{z_{\alpha/2}^2}{2}}\frac{C_1'}{\sqrt{n_0}}\alpha-r(\lambda)z_{\alpha/2}^2
$$
where $r(\lambda)$ is continuous in $\lambda$ and $r(\lambda)/(1-\lambda)\to 0$ as $\lambda\uparrow 1$. We know that $1-\lambda=C_1'/\sqrt{n_0}$ for some constant $C_1'>0$ from the choice of $\lambda$. For $n_0>32C'^2/\alpha^2$, $\lambda$ is bounded away from 0, and hence $|r(\lambda)/(1-\lambda)|=|\sqrt{n_0}r(\lambda)/C_1'|$ has a constant upper bound. Hence, there exists a constant $C_2'>0$ such that 
\begin{equation*}
    \left(1-\frac{z_{1-\alpha/2}^2}{z_{\lambda\alpha/2}^2}\right)\frac{z_{\lambda\alpha/2}^2n_0(N-n_0)}{N^3}\leq \frac{C_2'}{\sqrt{n_0}}\frac{z_{\lambda\alpha/2}^2n_0(N-n_0)}{N^3}.
\end{equation*}
Combining the results, we get that 
\begin{equation*}
    \hat{p}_{u,n_0}^{BE}-\hat{p}_{u,n_0}^{CLT}\leq \frac{C_2'}{\sqrt{n_0}}z_{1-\lambda\alpha/2}\sqrt{\frac{n_0(N-n_0)}{N^3}}\leq\frac{C_2'z_{1-\lambda\alpha/2}}{N}\leq \frac{C_3'}{N}
\end{equation*}
where $C_3'>0$ is a constant and we get the last inequality since $\lambda$ has a non-zero lower bound. 

Now notice that 
\begin{equation*}
    \hat{p}_{l,n_0}^{CLT}-\hat{p}_{l,n_0}^{BE}=\hat{p}_{u,n_0}^{BE}-\hat{p}_{u,n_0}^{CLT}+\frac{z_{\lambda\alpha/2}^2n_0}{N^2}
\end{equation*}
and 
\begin{equation*}
    \frac{z_{\lambda\alpha/2}^2n_0}{N^2}\leq \frac{z_{\lambda\alpha/2}^2}{N}=O(1/N).
\end{equation*}
Therefore, we could find a constant $C_0'$ such that the theorem holds.
\end{proof}


\acks 
\noindent We gratefully acknowledge support from the National Science Foundation under grants CAREER CMMI-1834710 and IIS-1849280. A preliminary conference version of this work has appeared in \cite{bai2020error}.

\bibliographystyle{APT}
\bibliography{bibliography}

\end{document}